\documentclass[a4paper, 11pt]{article}

\usepackage[utf8]{inputenc}
\usepackage[T1]{fontenc}   
\usepackage[hidelinks]{hyperref}
\usepackage{url}           
\usepackage{booktabs}       
\usepackage{amsfonts}      
\usepackage{nicefrac}     
\usepackage{microtype}    
\usepackage{xcolor}      
\usepackage{amsmath}
\usepackage[numbers,sort&compress]{natbib}
\usepackage{fullpage}
\usepackage{subcaption}

\usepackage{amsmath, amssymb, amsthm}
\usepackage{enumitem}
\usepackage{tikz} 
\usepackage[ruled,vlined,linesnumbered]{algorithm2e}
\usepackage{algpseudocode}
\usepackage{tikz}
\usepackage{amsmath}
\usepackage{xcolor}
\usepackage{mathtools}
\usepackage{algpseudocode}
\usepackage{thmtools, thm-restate}
\usepackage{subcaption}

\usetikzlibrary{positioning}
\usetikzlibrary{fit}
\usetikzlibrary{arrows.meta}

\algrenewcommand\algorithmicrequire{\textbf{Input:}}
\algrenewcommand\algorithmicensure{\textbf{Output:}}

\allowdisplaybreaks

\newcommand{\mmsf}{PMMS-feasible }

\newcommand{\efx}{{\normalfont\textsc{EFX}}}

\newtheorem{innercommon}{X}[section]
\newtheorem{observation}[innercommon]{Observation}
\newtheorem{definition}[innercommon]{Definition}
\newtheorem{lemma}[innercommon]{Lemma}
\newtheorem{example}[innercommon]{Example}

\newtheorem{claim}[innercommon]{Claim}
\newtheorem{remark}[innercommon]{Remark}
\newtheorem{proposition}[innercommon]{Proposition}

\newcommand\blfootnote[1]{%
  \begingroup
  \renewcommand\thefootnote{}%
  \NoHyper\footnote{#1}\endNoHyper
  \addtocounter{footnote}{-1}%
  \endgroup
}

\title{Epistemic Pairwise Maximin Share}

\author{ 
Michal Feldman$^\ast$
\quad
Amos Fiat$^\dagger$
\quad
Yael Nissan$^\mathsection$
\quad 
Tomasz Ponitka$^\ddagger$
}

\date{June 17, 2026}

\begin{document}

\maketitle

\blfootnote{
This project has been partially funded by the European Research Council (ERC) under the European Union's Horizon Europe Program (FACT, grant agreement No.~101170373), by an Amazon Research Award, by the NSF-BSF (grant number 2020788), and by the Israel Science Foundation Breakthrough Program (grant No.~2600/24).}
\blfootnote{$^\ast$Tel Aviv University and Microsoft ILDC, Israel. Email: \texttt{mfeldman@tauex.tau.ac.il}}
\blfootnote{$^\dagger$Tel Aviv University, Israel. Email: \texttt{fiat@tau.ac.il}}
\blfootnote{$^\ddagger$Tel Aviv University, Israel. Email: \texttt{yaelshemesh@mail.tau.ac.il}}
\blfootnote{$^\mathsection$Tel Aviv University, Israel. Email: \texttt{tomaszp@mail.tau.ac.il}}

\begin{abstract}
We introduce \emph{epistemic pairwise maximin share} (EPMMS), a new fairness notion for 
{fair division} of indivisible goods. 
{Two 
fundamental notions in this setting are envy-freeness up to any item (EFX) and pairwise maximin share (PMMS), with PMMS being stronger than EFX. While EFX has been extensively studied, far less is known about PMMS.}
Recent work shows that relaxing EFX via an epistemic perspective leads to substantial progress on the EFX problem, raising the question of whether a similar approach can advance our understanding of PMMS. Motivated by this, we initiate the study of EPMMS, the epistemic relaxation of PMMS.
EPMMS is 
more challenging than EEFX: the key approaches underlying recent progress on epistemic EFX inherently fail to extend to EPMMS.

We establish the following results. 
\begin{enumerate}
    \item For additive valuations, 
$4/5$-EPMMS allocations exist and can be efficiently computed.
\item For bivalued valuations, 
EPMMS allocations exist and can be efficiently computed; in fact, we obtain the stronger guarantee of epistemic groupwise maximin share (EGMMS), which also strengthens the existence of MMS allocations for this setting.
\item We prove that 
EPMMS allocations exist in two settings where MMS allocations need not exist: instances with three additive agents or two types of additive agents.
\end{enumerate}
\end{abstract}

\newpage
\clearpage

\section{Introduction}

The fair division problem asks how to allocate a set of items among {$n$} agents with differing preferences so that every agent is treated fairly. It arises in many settings, from allocating seats in university courses \cite{Budish10} to assigning papers to reviewers at conferences \cite{LianMNW18}. {Our primary focus is on the setting with \emph{indivisible} items (an allocation $X_1, \ldots, X_n$ forms a partition of the items into disjoint bundles) and \emph{additive} valuations (agent $i$'s value for a bundle $X$ is $v_i(X) = \sum_{g \in X} v_i(\{g\})$).}

For two agents, the canonical solution is the \emph{cut-and-choose} mechanism: one agent divides the items into two bundles, and the other picks their preferred one.
{With more than two agents, however, the situation becomes more complex, and no obvious canonical solution emerges.}
A central challenge in fair division is to identify properties of 
the cut-and-choose mechanism
that 
extend
beyond two agents.

Toward this goal, \citet*{CaragiannisKMPSW2019} proposed two fairness notions. 
The first is \emph{envy-freeness up to any good} (EFX). We say that agent $i$ envies agent $j$ if $i$ prefers $j$'s allocation to their own.
{Although cut-and-choose guarantees that the chooser does not envy the cutter, it is generally impossible to guarantee envy-freeness (EF) for both agents, as even the simple case of a single item demonstrates.}
{EFX is a relaxation of EF requiring that any}
envy disappears after hypothetically removing any single good from the envied bundle: $v_i(X_i) \geq v_i(X_j \setminus \{g\})$ for all $i,j \in [n]$ and all $g \in X_j$.

The second {fairness notion} is \emph{pairwise maximin share} (PMMS). Consider any pair of agents  $i$ and $j$, and imagine running cut-and-choose over their combined items $X_i \cup X_j$. The best guarantee $i$ could secure as the cutter is their \emph{maximin share},
$\max_{(A,B)} \min\{v_i(A), v_i(B)\}$, taken over all partitions $(A,B)$ of $X_i \cup X_j$. 
PMMS requires that every agent $i$ receives at least their
maximin share for the set $X_i \cup X_j$, for all $j\neq i$.
Under mild assumptions, PMMS is strictly stronger than EFX \cite{CaragiannisKMPSW2019,ByrkaMP25,GargS2025}.

Both EFX and PMMS are satisfied by cut-and-choose when there are only two agents, but whether either can be satisfied for an arbitrary number of agents with additive valuations is a long-standing open problem. The existence of EFX is known in important special cases: for three agents \cite{ChaudhuryGM24,AkramiACGMM2022}, for bivalued valuations \cite{AmanatidisBFHV20,JinT25,ByrkaMP25}, and for instances with few distinct valuation types \cite{Mahara2025,PrakashGNN2025}. 
Despite being introduced in the same paper as EFX, PMMS remains significantly less explored, and as the stronger notion, much less is known about its existence.

As a successful line of attack on the EFX problem, \citet*{CaragiannisGRSV25} introduced a compelling relaxation of EFX called \emph{Epistemic EFX} (EEFX). An allocation is EEFX if each agent 
{could hypothetically}
rearrange the items held by the other agents into bundles that, from their own perspective, constitute an EFX allocation. Notably, each agent can verify fairness on the basis of their own bundle alone, without reference to how the remaining items are actually partitioned among the other agents.

This relaxation has proven more tractable than EFX itself. \citet*{CaragiannisGRSV25} established the existence of EEFX allocations under additive valuations via a reduction to the identical-rankings setting, followed by an application of Envy Cycle Elimination algorithm. \citet*{AkramiN24} subsequently extended this result to arbitrary monotone valuations via the Lone Divider approach.

{Motivated by the substantial progress on EFX enabled by the epistemic approach, and by the comparatively limited understanding of PMMS,}
we initiate the formal study of \emph{Epistemic PMMS} (EPMMS), the analogous relaxation of PMMS. 
Notably, both of the two approaches underlying the existing EEFX results inherently fail to carry over to EPMMS,
as we demonstrate through explicit counterexamples in Appendix~\ref{sec:separations}.

Beyond comparison-based fairness notions like EFX and PMMS, a substantial body of work concerns \emph{share-based} fairness. The central share-based notion is the \emph{maximin share} (MMS), defined as the maximum value a cutter can guarantee by partitioning the entire set of items into $n$ bundles, and receiving the least valuable bundle. It is known that allocations giving every agent at least their MMS need not exist in general, even for three agents with additive valuations \cite{KurokawaPW18,FeigeST21}. The epistemic approach does not apply to MMS directly, as MMS is defined independently of how items are distributed among the remaining agents. Nonetheless, every MMS allocation satisfies EEFX \cite{CaragiannisGRSV25}. By contrast, in Appendix~\ref{sec:separations}, we show that MMS does not imply EPMMS, which yields a further separation between EEFX and EPMMS.
{See Figure~\ref{fig:implications} for a broader picture of these relationships.}

More broadly, a recurring theme in the literature on EFX and PMMS is the role of \emph{PMMS-feasibility}\footnote{PMMS-feasibility is exactly the condition introduced by~\citet*{AkramiACGMM2022} under the name MMS-feasibility. We refer to it as PMMS-feasibility for three reasons: first, to align the terminology with the GMMS-feasibility condition that we introduce in this paper (see Definition~\ref{def:gmmsfeasible}); second, to emphasize that PMMS-feasibility (resp.\ GMMS-feasibility) precisely characterizes the valuations for which envy-freeness implies PMMS (resp.\ GMMS), see Section~\ref{sec:ef_implies_pmms}; and third, because the original name can potentially be misleading: although PMMS-feasibility guarantees the existence of an MMS allocation for two agents (where MMS coincides with PMMS), no such guarantee holds for three or more agents.
},
a structural condition that generalizes additivity and also captures budget-additive, unit-demand, and multiplicative valuations \cite{AkramiACGMM2022}; see 
{Section~\ref{sec:model} and Appendix~\ref{sec:valuation_classes}.}
Without any assumptions, both notions face strong impossibilities: EFX can fail to exist under monotone valuations for three or more agents \cite{akrami2026counterexampleefxnge}, and PMMS fails to exist already for two agents \cite{ByrkaMP25}. PMMS-feasibility circumvents these impossibilities: EFX exists for three agents when at least one is PMMS-feasible \cite{AkramiACGMM2022}, and PMMS exists for two agents when one of them is PMMS-feasible \cite{ByrkaMP25}. The case of $n$ PMMS-feasible agents remains open for both notions. 
Notably, EEFX is guaranteed to exist for arbitrary monotone valuations \cite{AkramiN24}, whereas EPMMS inherits the non-existence of PMMS for two monotone agents, since the two notions coincide in this case. This marks yet another contrast between EEFX and EPMMS.

\subsection{Our Results}

We begin with additive valuations, for which we obtain a multiplicative approximation of EPMMS.

\begin{restatable}[Additive Valuations]{theorem}{thmadditive}
\label{thm:4_5_EPMMS}
Every instance with additive valuations admits a $4/5$-EPMMS allocation.
Moreover, this allocation can be computed in polynomial time.
\end{restatable}

To establish the existence of $4/5$-EPMMS for additive valuations, we prove that $4/5$-PMMS allocations exist under the \emph{identical-rankings} assumption\footnote{Both our work and that of \citet*{DaiGMGXZ2024} claim the existence of a $4/5$-PMMS allocation for instances with identical rankings via the Envy Cycle Elimination algorithm. However, the proof of \citet{DaiGMGXZ2024} contains issues that, to the best of our knowledge, remain unresolved. Crucially, they do not argue that an envy cycle elimination step (Line~\ref{line:envy_cycle} of Algorithm~\ref{alg:envy_cycles}) preserves the $4/5$-PMMS guarantee. This does not follow automatically: in Appendix~\ref{pmmsincycleelimination}, we show that in general a single envy cycle elimination step can only guarantee $2/3$-PMMS, even when starting from an exact PMMS allocation. We show that the $4/5$-PMMS guarantee can nevertheless be recovered for the Envy Cycle Elimination algorithm in the setting with identical rankings, but only under the additional assumption that the valuations are non-degenerate, which is missing from the analysis of \citet{DaiGMGXZ2024}. In Appendix~\ref{sec:envy_cycles_degenerate}, we provide examples showing that without non-degeneracy, the Envy Cycle Elimination algorithm does \emph{not} guarantee $4/5$-PMMS, even under identical rankings and even in the absence of zero-valued items.
\label{footnote_4_5_pmms}
}.
For comparison, \citet*{Kurokawa17} proved the existence of $0.781$-PMMS allocations for additive valuations; our result shows that the epistemic  relaxation improves this factor to $0.8$ {(but for the weaker notion of EPMMS)}.

We next establish the existence of exact EPMMS allocations in several important settings. While EFX is known to exist in each of these settings \cite{AmanatidisBFHV20,JinT25,ByrkaMP25,ChaudhuryGM24,AkramiACGMM2022,Mahara2025,PrakashGNN2025}, EPMMS and EFX are incomparable notions, so our results neither imply nor are implied by these EFX existence results.

First, we consider \emph{bivalued} valuations, in which each agent $i$'s valuation function $v_i$ is additive and satisfies $v_i(\{g\}) \in \{a, b\}$ for every item $g$ and some fixed $a, b \geq 0$. In this setting, we in fact establish a stronger guarantee: the existence of an \emph{Epistemic GMMS} (EGMMS) allocation. The groupwise maximin share (GMMS) criterion strengthens both PMMS and MMS by requiring that each agent receive at least their maximin share within every subgroup of agents to which they belong; EGMMS is the analogous epistemic relaxation (see Definition~\ref{def:egmms}).

\begin{restatable}[Bivalued Valuations]{theorem}{thmbivalued}
\label{thm:GMMS_bivalued}
Every instance with bivalued valuations admits an EGMMS, and hence an EPMMS, allocation. Moreover, this allocation can be computed in polynomial time.\end{restatable}

In particular, our result strengthens that of \citet*{Feige2022MMSab}, who proved the existence of MMS allocations for bivalued instances, since EGMMS implies MMS. We note that EGMMS is a strictly stronger notion: in fact, even PMMS and MMS simultaneously do not imply EGMMS; see Section~\ref{sec:sep_egmms_epmms}.  

We now turn to the case of three agents. For this case, we additionally guarantee that the allocation is individually MMS-satisfying or PMMS-satisfying (IMMP), which requires that each agent either (1) receives at least her MMS share, or (2) does not PMMS-envy any other agent.

\begin{restatable}[Three Agents]{theorem}{thmthreeagents}\label{thm:threeagents}
Every instance with three additive agents admits an allocation that is simultaneously EPMMS and IMMP. More generally, (a) to achieve EPMMS it suffices that at least two agents are PMMS-feasible, (b) to achieve both EPMMS and IMMP it suffices that one of the agents is PMMS-feasible and another agent is GMMS-feasible.
\end{restatable}

The IMMP property
is a natural generalization of the {individually MMS-satisfying or EFX-satisfying} (IMMX) property introduced by \citet*{babaioff2026nearoptimalbestofbothworldsfairnessagents}, 
which requires that each agent either (1) receives at least her MMS share, or (2) does not EFX-envy any other agent.
While the existence of IMMX for three additive agents follows from the existence of EFX in this setting \cite{ChaudhuryGM24}, the same reasoning does not apply to IMMP, since the existence of PMMS for three additive agents remains open.

{Finally}, we consider instances with two \emph{types} of agents, meaning that the agents are partitioned into two groups within which all members share the same valuation function.

\begin{restatable}[Two Types of Agents]{theorem}{thmtwotypes}
\label{thm:two_types_pmms}
Every instance with two types of additive agents admits an EPMMS allocation.
More generally, this holds whenever at least one of the types is PMMS-feasible.
\end{restatable}

For both of the last two settings, MMS allocations need not exist~\citep{FeigeST21}. Consequently, since EGMMS implies MMS, neither theorem can be strengthened to EGMMS in general.

\subsection{Related Works}
\label{sec:related}

In this subsection, we summarize the literature most closely related to our work. For a broader overview of the field of fair division, see the survey of \citet*{AmanatidisABFLMVW23}.

\paragraph{Epistemic-EFX.}
The epistemic perspective on fairness was initiated by \citet*{AzizBCGL18}, who introduced \emph{epistemic envy-freeness}. Building on this idea, \citet*{CaragiannisGRSV25} introduced the analogous relaxation of EFX, namely EEFX, and established its existence for additive valuations. \citet*{AkramiN24} subsequently extended the existence result to arbitrary monotone valuations. More recently, \citet*{AkramiMMR26} strengthened these guarantees by showing that EEFX can be achieved simultaneously with envy-freeness up to some item (EF1). 

\paragraph{PMMS.} \citet*{CaragiannisKMPSW2019} showed that Nash-welfare-maximizing allocations are $0.618$-PMMS. The best known approximation of PMMS for additive valuations is $0.781$, due to \citet*{Kurokawa17}. Exact PMMS allocations exist in several restricted settings: \citet*{ByrkaMP25} established existence for binary-valued and pair-demand valuations, and \citet*{ChristodoulouM26} showed existence under additive valuations whenever each item has non-negative value for at most two agents. PMMS has also been studied alongside other fairness notions: \citet*{AmanatidisBM2018} showed that EFX implies $2/3$-PMMS and EF1 implies $1/2$-PMMS; \citet*{AmanatidisNM20} obtained allocations that are simultaneously $0.618$-EFX and $0.717$-PMMS; and \citet*{FeldmanMP24} achieved $0.618$-EFX, $0.618$-PMMS, and a $0.618$-fraction of the maximum Nash welfare simultaneously.

\paragraph{GMMS.} The GMMS notion was introduced by \citet*{BarmanBKN2017}. For additive valuations, the best known multiplicative approximation is $4/7$, established independently by \citet*{AmanatidisNM20} and \citet*{chaudhuryKMS20}. Several works have further studied GMMS in conjunction with other fairness notions: \citet*{AmanatidisNM20} established the existence of $0.55$-GMMS and $0.618$-EFX allocations, while \citet*{FeldmanMP24} obtained allocations that are simultaneously $0.44$-GMMS, $0.618$-EFX, and achieve a $0.618$-fraction of the maximum Nash welfare. Beyond additive valuations, \citet*{BarmanV21} established a $1/6$-GMMS approximation for XOS valuations.

\paragraph{Bivalued Valuations.}
The class of bivalued valuations, where $v_i(\{g\}) \in \{a,b\}$ for every agent $i$ and item $g$, has received significant attention as a natural restriction under which several fairness notions become tractable. \citet*{Feige2022MMSab} established the existence of MMS allocations in this setting, in contrast to the additive case \cite{FeigeST21,KurokawaPW16}. \citet*{AmanatidisBFHV20} showed that EFX allocations always exist under bivalued valuations, and this was extended to personalized bivalued valuations, where $v_i(\{g\}) \in \{a_i,b_i\}$, by \citet*{ByrkaMP25} and \citet*{JinT25}. For the important special case of binary additive valuations ($v_i(\{g\}) \in \{0,1\}$), \citet*{ByrkaMP25} further showed that PMMS allocations exist. 

\paragraph{Few Agents.} A natural line of research approaches the EFX existence problem by restricting the number of agents. The case of two agents was settled by \citet*{PlautR17}, who showed that EFX allocations always exist; for three agents, existence was established by \citet*{ChaudhuryGM24} for additive valuations and subsequently extended to PMMS-feasible valuations by \citet*{AkramiACGMM2022}. For four agents, \citet*{BergerCFF2021} showed that EFX allocations exist provided that one item can be left unallocated, and \citet*{AshuriGS25} established the existence of EF2X allocations, a relaxation of EFX in which envy must vanish after removing any two items from the envied bundle. 
 
\paragraph{Few Types of Agents.} A complementary line of research restricts the number of distinct valuation functions rather than the number of agents. \citet*{Mahara2025} established the existence of EFX allocations when the agents partition into at most two types, and \citet*{PrakashGNN2025} extended this to three types. For an arbitrary number of types $k$, \citet*{GhosalHN025} showed that EFX allocations exist provided that $k-2$ items can be left unallocated.

\section{Model and Preliminaries}
\label{sec:model}

\begin{figure}[t]
\centering

\begin{subfigure}[b]{\textwidth}
    \centering
    \begin{tikzpicture}[
        scale=1,
        transform shape,
        box/.style={draw,rectangle,minimum height=0.8cm,minimum width=1.8cm},
        arrshort/.style={-{Latex[length=1.5mm,width=1mm]},shorten >=1pt},
        arrlong/.style={-{Latex[length=1.5mm,width=1mm]}}
    ]
    \node[box] (gmms)   at (0,0)            {GMMS};
    \node[box] (pmms)   at (2.8,1)        {PMMS};
    \node[box] (efx)    at (5.6,2)        {EFX};
    \node[box] (egmms)  at (2.8,-1)       {EGMMS};
    \node[box] (epmms)  at (5.6,0)          {EPMMS};
    \node[box] (eefx)   at (8.4,1)        {EEFX};
    \node[box] (mms)    at (8.4,-1)      {MMS};
    \node[box] (approx) at (11.2,0)      {$4/7$-MMS};
    
    \draw[arrshort] (gmms) -- (pmms);
    \draw[arrshort] (gmms) -- (egmms);
    \draw[arrlong] (egmms) -- (mms);
    \draw[arrshort] (egmms) -- (epmms);
    \draw[arrshort] (pmms) -- (epmms);
    \draw[arrshort] (pmms) -- (efx);
    \draw[arrshort] (epmms) -- (eefx);
    \draw[arrshort] (efx) -- (eefx);
    \draw[arrlong] (mms) -- (eefx);
    \draw[arrshort] (eefx) -- (approx);
    \draw[arrshort] (mms) -- (approx);
    
    \draw[{Latex[length=1.5mm,width=1mm]}-{Latex[length=1.5mm,width=1mm]}, gray, shorten >=1pt, shorten <=1pt] (epmms) -- (mms) 
        node[midway, sloped, font=\Large\bfseries] {$\times$};
    \draw[{Latex[length=1.5mm,width=1mm]}-{Latex[length=1.5mm,width=1mm]}, gray, shorten >=1pt, shorten <=1pt] (epmms) -- (efx) 
        node[midway, sloped, font=\Large\bfseries] {$\times$};
    \end{tikzpicture}
    \caption{Implications involving the epistemic fairness notions.}
\end{subfigure}

\vspace{1.5em} %

\begin{subfigure}[b]{\textwidth}
    \centering
    \begin{tikzpicture}[
        scale=1,
        transform shape,
        box/.style={draw,rectangle,minimum height=0.8cm,minimum width=1.8cm},
        arrshort/.style={-{Latex[length=1.5mm,width=1mm]},shorten >=1pt},
        arrlong/.style={-{Latex[length=1.5mm,width=1mm]}}
    ]
    \node[box] (pmms)   at (-1.5, 1.5)   {PMMS};
    \node[box] (efx)    at (1.5, 1.5)    {EFX};
    
    \node[box] (mms)    at (-4.5, 0)   {MMS};
    \node[box] (immp)   at (-1.5, 0)   {IMMP};
    \node[box] (immx)   at (1.5, 0)    {IMMX};
    \node[box] (approx) at (4.5, 0)    {$4/7$-MMS};
    
    \draw[arrshort] (pmms) -- (efx);
    
    \draw[arrshort] (mms) -- (immp);
    \draw[arrshort] (immp) -- (immx);
    \draw[arrshort] (immx) -- (approx);
    
    \draw[arrshort] (pmms) -- (immp);
    \draw[arrshort] (efx) -- (immx);
    \end{tikzpicture}
    \caption{Implications involving IMMP and IMMX.}
\end{subfigure}

\caption{Implications between fairness notions under non-degenerate additive valuations. Most implications follow directly from the definitions; the following require proof: (1)~PMMS $\Rightarrow$ EFX and EPMMS $\Rightarrow$ EEFX on non-degenerate instances~\cite{CaragiannisKMPSW2019}; (2)~MMS $\Rightarrow$ EEFX~\cite{CaragiannisGRSV25}; and (3)~EEFX and IMMX $\Rightarrow$ $4/7$-MMS~\cite{AmanatidisBM2018}. In Appendix~\ref{sec:mms_and_pmms}, we show that EPMMS is incomparable to both MMS and EFX.}
\label{fig:implications}
\end{figure}
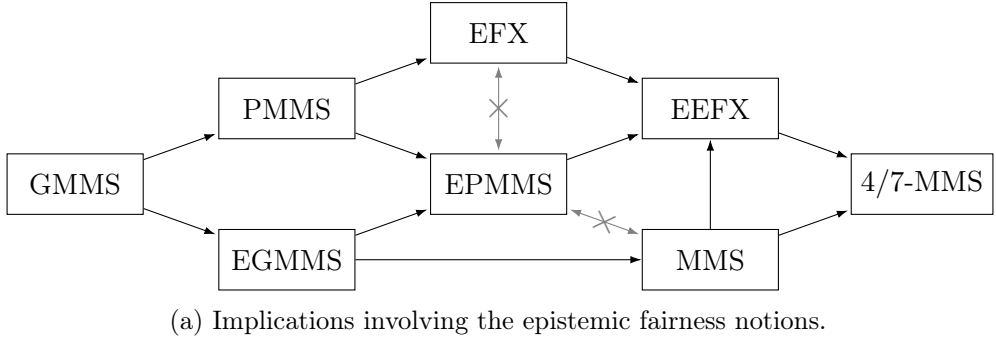
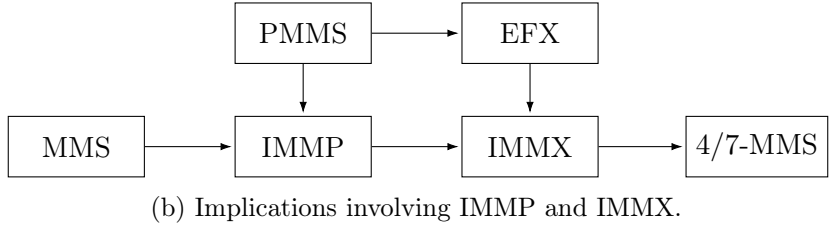

We consider a set of $n$ agents, $N = \{1,2,\ldots,n\}$, and a set $M$ of $m$ indivisible goods. Each agent $i \in N$ is associated with a monotone valuation function 
{$v_i : 2^M \to \mathbb{R}_{\geq 0}$}.
We will specify our assumptions on the valuation classes later. For notational convenience, we write $v_i(g)$ instead of $v_i(\{g\})$ for each $g \in M$.

Our goal is to study allocations that assign all goods in $M$ to the agents. Formally, such an allocation is a partition $X = (X_1, X_2, \ldots, X_n)$ of $M$, where $X_i \cap X_j = \emptyset$ for all $i \ne j$, and $\bigcup_{i \in N} X_i = M$, and $X_i$
is the set of goods allocated to agent $i$. We denote by $\Pi_n(S)$ the set of all partitions of a set $S$ into $n$ bundles.
We now introduce the fairness notions, beginning with the maximin share.

\begin{definition}[Maximin Share, \cite{Budish10}]
    Given n agents and a subset $S\subseteq M$ of goods, the n-maximin share of agent $i$ with respect to $S$ is defined as:
    $
    \mu_i(n, S) = \max_{X \in \Pi_n(S)} \min_{X_j \in X} v_i (X_j)
    $
\end{definition}

\subsection{Fairness Notions}
We now define the following fairness notions.
We refer to Figure~\ref{fig:implications} for an illustration of the implications between those fairness notions.

\begin{definition}[$\alpha$-EFX, \cite{CaragiannisKMPSW2019}] An allocation $X = (X_1, ..., X_n)$ is 
$\alpha$-envy-free up to any good ($\alpha$-EFX)
if, for any pair of agents $i,j \in N$, and any item $g \in X_j$, it holds that
   $v_i(X_i) \ge \alpha v_i(X_j \setminus \{g\})$.

\end{definition}

\begin{definition}[$\alpha$-MMS, \cite{Budish10}]
\label{def:MMS}
    Let $\alpha \in [0,1]$. An allocation $X = (X_1,..., X_n)$ is 
    $\alpha$-maximin-share-fair ($\alpha$-MMS)
    if, for every $i \in N$, it holds that $v_i(X_i) \ge \alpha \mu_i(n,M)$.
\end{definition}

\begin{definition}[$\alpha$-PMMS, \cite{CaragiannisKMPSW2019}]
\label{def:PMMS}
 Let $\alpha \in [0,1]$. An allocation $X$ is 
 $\alpha$-pairwise-maximin-share-fair 
($\alpha$-PMMS)
if, for any pair of agents $i,j \in N$, it holds that
$v_i(X_i) \ge \alpha \mu_i(2, X_i \cup X_j).$
\end{definition}

{Whenever this inequality is violated for agents $i$ and $j$, we say that agent $i$ $\alpha$-PMMS-envies agent $j$.}

\begin{definition}[$\alpha$-GMMS, \cite{BarmanBKN2017}]
An allocation $X$ is a $\alpha$-groupwise-maximin-share-fair ($\alpha$-GMMS)  if, for all $J \subseteq [n]$ with $i\in J$, it holds that
$v_i(X_i) \ge \alpha \mu_i(|J|, \cup_{j\in J} X_j)$.

\end{definition}
Note that GMMS strictly generalizes both MMS and PMMS \cite{BarmanBKN2017}. {Moreover, for non-degenerate valuations (i.e., valuations under which no two bundles share the same value), PMMS implies EFX \cite{CaragiannisKMPSW2019,ByrkaMP25,GargS2025}.} Next, we define the epistemic variants of the fairness notions above. The notion of EEFX has been defined by \cite{CaragiannisGRSV25}, and the notions of EPMMS and EGMMS are introduced herein.

\begin{definition}[EEFX, \cite{CaragiannisGRSV25}] 
An allocation $X = (X_1, X_2, \ldots, X_n)$ is {\em Epistemic $\alpha$-\efx\ ($\alpha$-EEFX)}
if for every $i \in [n]$, there exists an allocation $Y^i = (Y^i_1, Y^i_2, \ldots, Y^i_n)$ such that
$Y^i_i = X_i$ and agent $i$ is $\alpha$-EFX-satisfied with allocation $Y^i$, i.e., $v_i(Y_i^i) \geq \alpha v_i(Y^i_j \setminus \{g\})$ for all $j$ and $g \in Y^i_j$.
\end{definition}

\begin{definition}[EPMMS] 
An allocation $X = (X_1, X_2, \ldots, X_n)$ is {\em Epistemic $\alpha$-PMMS ($\alpha$-EPMMS)}
if for every  $i \in [n]$, there exists an allocation $Y^i = (Y^i_1, Y^i_2, \ldots, Y^i_n)$ such that
$Y^i_i = X_i$ and agent $i$ is $\alpha$-PMMS-satisfied with allocation $Y^i$, i.e., $v_i(Y_i^i) \geq \alpha  \mu_i(2, Y^i_i \cup Y^i_j)$ for all $j$.
\end{definition}

\begin{definition}[EGMMS]\label{def:egmms}
An allocation $X = (X_1, X_2, \ldots, X_n)$ is {\em Epistemic $\alpha$-GMMS ($\alpha$-EGMMS)}
if for every $i \in [n]$, there exists an allocation $Y^i = (Y^i_1, Y^i_2, \ldots, Y^i_n)$ such that
$Y^i_i = X_i$ and agent $i$ is $\alpha$-GMMS-satisfied with $Y^i$, i.e., $v_i(X_i) \ge \alpha \mu_i(|J|, \cup_{j\in J} Y^i_j)$ for all $i \in J \subseteq [n]$.
\end{definition}

For each notion above, the allocation $Y^i$ is called a \emph{certificate} of agent $i$ for bundle $X_i$.
When $\alpha = 1$, we omit it from the notation, writing, e.g., EPMMS rather than $1$-EPMMS.

We also consider and extend fairness notions that are of the form ``for each agent individualy, either {\sl this} or {\sl that} holds''. 
The first such notion, of either MMS or EFX, was introduced by
\citet{babaioff2026nearoptimalbestofbothworldsfairnessagents}.
We introduce the notion of either MMS or PMMS.

\begin{definition}[IMMX, \cite{babaioff2026nearoptimalbestofbothworldsfairnessagents}]\label{def:immx}
    An allocation $X = (X_1, \ldots, X_n)$ is \emph{Individually MMS-Satisfying or EFX-Satisfying} (IMMX) if, for each agent $i$, at least one of the following holds: (1) agent $i$ receives at least her MMS share, i.e.,
    $v_i(X_i) \geq \mu_i(n, M)$; or (2) agent $i$ does not EFX-envy any other agent, i.e., $v_i(X_i) \geq v_i(X_j \setminus \{g\})$ for all $j \neq i$ and all $g \in X_j$.
\end{definition}

\begin{definition}[IMMP]\label{def:immp}
    An allocation $X = (X_1, \ldots, X_n)$ is \emph{Individually MMS-Satisfying or PMMS-Satisfying} (IMMP) if, for each agent $i$, at least one of the following holds: (1) agent $i$ receives at least her MMS share, i.e., $v_i(X_i) \geq \mu_i(n, M)$; or (2) agent $i$ does not PMMS-envy any other agent, i.e., $v_i(X_i) \geq \mu_i(2, X_i \cup X_j)$ for all $j \neq i$.
\end{definition}

\subsection{Valuation Classes} 
We next define the valuation classes considered in this work, beginning with the standard class of additive valuations.

\begin{definition}[Additive Valuations]
    A valuation function $v: 2^M \rightarrow \mathbb{R}_{\geq 0}$ is additive if for every subset of goods $S \subseteq M$, we have $v_i(S) = \sum_{g \in S} v_i(g)$. 
\end{definition}

Beyond additive valuations, we consider two natural generalizations: the class of PMMS-feasible valuations, introduced by \cite{AkramiACGMM2022}, and the class of GMMS-feasible valuations, which we introduce here.

\begin{definition}[PMMS-feasible Valuations, \cite{AkramiACGMM2022}]\label{def:mmsfeasible}
    A valuation function $v: 2^M \rightarrow \mathbb{R}_{\geq 0}$ is \mmsf if for every subset of goods $S \subseteq M$ and partitions $A = (A_1, A_2)$ and $B = (B_1, B_2)$ of $S$, we have
$\max\{ v(B_1), v(B_2) \} \geq \min \{ v(A_1), v(A_2) \}$. 
\end{definition}

\begin{definition}[GMMS-Feasible Valuations]\label{def:gmmsfeasible}
    {A valuation function $v: 2^M \rightarrow \mathbb{R}_{\geq 0}$ is GMMS-feasible if for any subset of goods $S \subseteq M$, every $k \le m$, and any pair of partitions $A = (A_1, \ldots, A_k)$ and $B = (B_1, \ldots, B_k)$ of $S$, we have
        $\max \{ v(B_1), \ldots, v(B_k) \}  \geq \min \{ v(A_1), \ldots, v(A_k) \}$.}
\end{definition}
\section{Structural Insights}

\subsection{Relations Between Valuation Classes}

\begin{figure}[t]
\centering
\begin{tikzpicture}[scale=1]
\fill[gray!10] (2,0) ellipse (3.9 and 2.3);
\draw[thick] (2.75,0) ellipse (4.75 and 2.55);
\node[scale=0.9] at (6.6, 0.2)  {PMMS-};
\node[scale=0.9] at (6.6,-0.2)  {Feasible};
\draw[thick] (-2.5,0) ellipse (3.5 and 2.5);
\node[scale=0.9] at (-5, 0.2)  {Gross};
\node[scale=0.9] at (-5,-0.2)  {Substitutes};
\draw[thick] (-1.5,0) ellipse (2.4 and 2);
\node[scale=0.9] at (-3, 0) {OXS};
\draw[thick] (2,0) ellipse (3.9 and 2.3);
\node[scale=0.9] at (4.5, 0.2)  {GMMS-};
\node[scale=0.9] at (4.5,-0.2)  {Feasible};
\draw[thick] (0.8,0) ellipse (2.6 and 1.8);
\node[scale=0.9] at (1.8,-0.6) {Nice};
\node[scale=0.9] at (1.8,-1) {Cancelable};
\draw[thick] (0.6,0.55) ellipse (2.0 and 0.9);
\node[scale=0.9] at (1.7,0.75) {Budget};
\node[scale=0.9] at (1.7,0.35) {Additive};
\draw[thick] (-0.4,0.55) ellipse (0.85 and 0.45);
\node[scale=0.9] at (-0.4,0.55) {Additive};
\draw[thick] (-0.4,-0.85) ellipse (0.8 and 0.5);
\node[scale=0.8] at (-0.4,-0.73) {Unit-};
\node[scale=0.8] at (-0.4,-0.97) {Demand};
\end{tikzpicture}
\caption{An illustration of the relationships between the valuation classes derived in Appendix~\ref{sec:valuation_classes}.}
\label{fig:gsandmmsfeasible}
\end{figure}
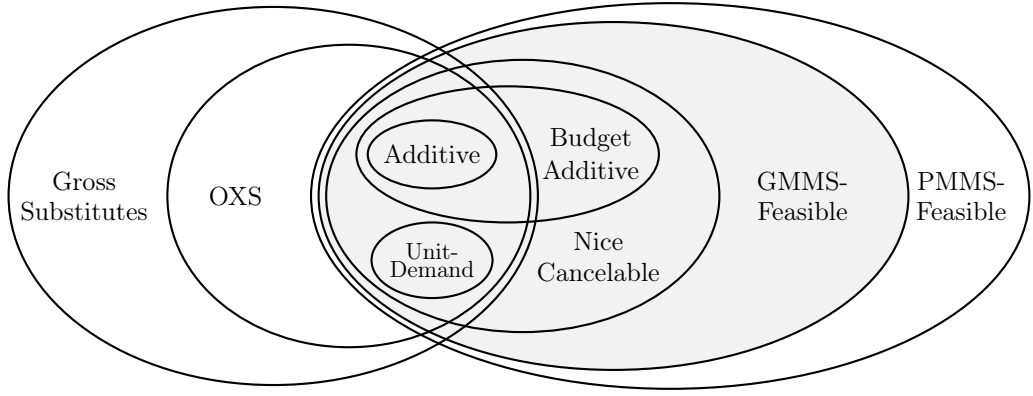

In Appendix~\ref{sec:valuation_classes}, we establish several structural properties of GMMS-feasible valuations. 
\begin{enumerate}
    \item GMMS-feasible valuations form a \emph{strict} subclass of PMMS-feasible valuations.
    \item  GMMS-feasible valuations form a \emph{strict} superclass of nice cancelable valuations \cite{BergerCFF2021}, which contain all additive, budget-additive, unit-demand, and multiplicative valuations.
    \item PMMS-feasibility (resp.\ GMMS-feasibility) 
precisely characterizes valuations where envy-freeness implies PMMS (resp.\ GMMS).
\item  The class of gross substitutes (GS) valuations is incomparable with both MMS- and GMMS-feasible valuations: there exist GS valuations that are not PMMS-feasible, and GMMS-feasible valuations that are not GS.
\end{enumerate}
See Figure~\ref{fig:gsandmmsfeasible} for an illustration of these relations.

\subsection{The Leximin Solution}
The leximin (lexicographic maximin) solution is a foundational concept in fair division \cite{rawls1971theory,PlautR17} that selects the allocation which maximizes the minimum utility;
then, if there are multiple allocations which achieve that minimum utility, it chooses among those
the allocation which maximizes the second minimum utility, and so on. We adopt the following formal definition.

\begin{definition}[Leximin Solution, \cite{PlautR17}]
Fix a valuation function $v$.
Let $\Pi_n(M)$ be the set of all possible allocations.
For any $X \in \Pi_n(M)$, let $v(X) = (v(X_1), \dots, v(X_n))$ be the utility profile, and let $v^{\uparrow}(X)$ be this profile sorted in non-decreasing order. An allocation $X$ is a leximin solution under $v$ if it lexicographically maximizes $v^{\uparrow}(X)$ over all $X \in \Pi_n(M)$. That is, $X$ is strictly preferred to $Y$ if at the first index $k$ where $v^{\uparrow}_k(X) \neq v^{\uparrow}_k(Y)$, we have $v^{\uparrow}_k(X) > v^{\uparrow}_k(Y)$.
\end{definition}
\citet*{BarmanBKN2017} showed that any leximin allocation is GMMS when all agents have identical valuations. Since every GMMS allocation is, by definition, both a PMMS and an MMS allocation, it naturally follows that the leximin solution guarantees both properties for identical-valuation instances.
\begin{proposition}[\cite{BarmanBKN2017}] \label{prop:leximin_is_gmms}
    For instances where all agents have identical monotone valuations (i.e., $v_i = v$ for all $i \in N$ and some $v$), any leximin allocation under $v$ is a GMMS allocation.
\end{proposition}

\subsection{Identical-Ranking Instances}
We now show how to reduce instances with arbitrary rankings over items to instances in which all agents share an identical ranking.

\begin{definition}[Identical-Ranking Instances] An instance has identical rankings (or is an ordered instance) if there exists a common ordering of the goods $g_1, g_2, \ldots, g_m$ such that for every agent $i \in N$, it holds that 
    $v_i(g_1) \ge v_i(g_2) \ge \ldots \ge v_i(g_m)$.
\end{definition}

{The reduction to identical rankings is a powerful tool in fair division, as it enables one to establish existence results on a significantly simpler instance. This perspective has been particularly useful in the study of Maximin Share (MMS) allocations, where it has been leveraged to compute $2/3$-MMS allocations \cite{AmanatidisMNS17}, as well as several other MMS approximation results \cite{BouveretL16,KurokawaPW18}. More recently, \cite{CaragiannisGRSV25} employed this framework to prove the existence of Epistemic EFX (EEFX) allocations for additive valuations. In this paper, we build on this approach to obtain guarantees for EPMMS approximation under additive valuations and EGMMS allocations under bivalued valuations.}

\begin{proposition}[Reduction to Identical Rankings]
\label{prop:EPMMS_EGMMS_identical_rankings_WLOG}
Let $0 < \alpha \le 1$.
\begin{itemize}
    \item If an $\alpha$-EPMMS allocation exists for all additive instances with identical rankings, then an $\alpha$-EPMMS  allocation exists for all additive instances with arbitrary rankings.
    \item If an $\alpha$-EGMMS allocation exists for all bivalued instances with identical rankings, then an $\alpha$-EGMMS  allocation exists for all bivalued instances with arbitrary rankings.
\end{itemize}
{Moreover, in both cases, if the former allocations can be computed in polynomial time, then so can the latter}.
\end{proposition}

{It is worth noting that the reduction for $\alpha$-EGMMS extends to additive valuations as well. We state it here for bivalued instances, as this is the setting for which we establish EGMMS guarantees.}
{While the reduction algorithm itself is essentially identical to that of \cite{CaragiannisGRSV25} for EEFX (and similarly to \cite{AmanatidisMNS17} for $2/3$-MMS), the analysis is more involved. Unlike these settings, where one reasons about the removal of a single item from another agent's bundle, here we must handle substantially richer deviations: arbitrary repartitions of pairs of bundles for EPMMS, and combinations of multiple bundles for EGMMS. These features introduce new technical challenges that require a more delicate argument. The proof of Proposition~\ref{prop:EPMMS_EGMMS_identical_rankings_WLOG} and  the presentation of the ordering framework appear in Section~\ref{sec:EPMMS_EGMMS_Reductions_framework}.}

\section{Approximating EPMMS for Additive Valuations}\label{sec:additivew}
In this section, we show that 4/5-EPMMS allocation exists for additive valuations. 

\thmadditive*

In fact, our proof yields a stronger guarantee: an allocation that is simultaneously $4/5$-EPMMS and EEFX. 
By the reduction in Proposition~\ref{prop:EPMMS_EGMMS_identical_rankings_WLOG}, it suffices to prove the following proposition in order to establish Theorem~\ref{thm:4_5_EPMMS}.

\begin{restatable}[Additive Valuations with Identical Rankings]{proposition}{propositionidenticalrankings}
\label{prop:4_5_PMMS_identical_rankings_1}
Every instance with additive valuations and identical rankings admits a $4/5$-PMMS and EFX allocation. Moreover, this allocation  can be computed in polynomial time.
\end{restatable}

We refer the reader to Footnote~\ref{footnote_4_5_pmms} for previous work related to Proposition~\ref{prop:4_5_PMMS_identical_rankings_1}.

Before proving Proposition~\ref{prop:4_5_PMMS_identical_rankings_1}, we state a useful lemma.

\begin{restatable}{lemma}{lemmamms}
\label{lemmamms}
Fix agent $i$, and let $P, Q$ be two disjoint sets of items. Then, for any $S \subseteq P \cup Q$, 
\[
\mu_i(2, P \cup Q) \le \mu_i(2, (P \cup Q) \setminus S) + v_i(S).
\]
\end{restatable}

\begin{proof}
Assume that $(A_1, A_2)$ is a partition of $P \cup Q$ such that $\min \{ v(A_1), v(A_2) \} = \mu_i(2, P \cup Q)$.  
Without loss of generality, assume that $v_i(A_1 \setminus S) \le v_i(A_2 \setminus S)$.  
Since $(A_1 \setminus S, A_2 \setminus S)$ is a partition of $(P \cup Q) \setminus S$, it follows that 
$v_i(A_1 \setminus S) \le \mu_i(2, (P \cup Q) \setminus S)$. Thus,
\[
\mu_i(2, P \cup Q) \le v_i(A_1) \le v_i(A_1 \setminus S) + v_i(S) \le \mu_i(2, (P \cup Q) \setminus S) + v_i(S). \qedhere
\]
\end{proof}

We are now ready to prove Proposition~\ref{prop:4_5_PMMS_identical_rankings_1}.

\begin{algorithm}[t]
\caption{Envy Cycle Elimination Algorithm \cite{LiptonMMS04}}
\label{alg:envy_cycles}
\KwIn{An instance with non-degenerate additive valuations and identical rankings.}
\KwOut{A $4/5$-PMMS and EFX allocation.}
Initialize $X \gets (\emptyset, \ldots, \emptyset)$ and $U \gets M$\;
\While{$U \neq \emptyset$}{
    \eIf{$\exists\, i \in [n]$ such that $v_j(X_j) \geq v_j(X_i)$ for all $j \in [n]$}{
        Let $g \in \arg\max_{g' \in U} v_i(g')$ be the most valuable item in $U$\;
        $X_i \gets X_i \cup \{g\}$ and $U \gets U \setminus \{g\}$\tcp*{Assign $g$ to an unenvied agent.}\label{line:add_item}
    }{
        Let $(i_1, \ldots, i_\ell)$ satisfy \ $v_{i_r}(X_{i_{r+1}}) > v_{i_r}(X_{i_r})$ $\,\forall r \in [\ell]$,
        with indices taken mod $\ell$\;
        $X_{i_r} \gets X_{i_{r+1}}$ for all $r \in [\ell]$\tcp*{Eliminate an envy cycle.}\label{line:envy_cycle}
    }
}
\Return $X$\;
\end{algorithm}

\begin{proof}[Proof of Proposition~\ref{prop:4_5_PMMS_identical_rankings_1} ]
    We assume that the valuations are non-degenerate, meaning that $v_i(S) \neq v_i(T)$ for every agent $i$ and every pair of distinct bundles $S,T$. This assumption is without loss of generality; see Proposition~\ref{prop:nd_wlog} for a formal argument.

        We claim that the Envy Cycle Elimination algorithm (Algorithm~\ref{alg:envy_cycles}, \cite{LiptonMMS04}) produces an $4/5$-PMMS and EFX allocation in polynomial time. Polynomial-time termination was established by \citet*{LiptonMMS04}, and \citet*{PlautR17} showed that the resulting allocation satisfies EFX; it therefore remains to prove that it also satisfies  $4/5$-PMMS.

    {Let $X$ be an allocation produced by the above procedure, and let}
    $i,j \in N$. We show that 
\begin{equation}
\label{eq:4-5-bound}
    v_i(X_i) \ge \tfrac{4}{5}\mu_i(2, X_i \cup X_j).
\end{equation}

\textbf{Case 1:} $|X_j| = 1$. 
    Denote $X_j = \{g\}$. Using Lemma~\ref{lemmamms}
    we obtain
    \[
     \mu_i(2, X_i \cup X_j) \le \mu_i(2, \{g\}) + v_i(X_i) = v_i(X_i),
     \]
    {implying Equation~\eqref{eq:4-5-bound}, as desired}.

\textbf{Case 2:} $|X_j| \geq 2$. 
Consider the point during the execution of Algorithm~\ref{alg:envy_cycles} at which bundle $X_j$ is formed: some agent $j'$ holds a bundle $\tilde{X}_{j'} = X_j \setminus \{g\}$, and item $g$ is added to it in Line~\ref{line:add_item}, producing $X_j$. Note that $j'$ need not equal $j$, since the bundle may be reassigned to agent $j$ at a later stage via the envy-cycle operation in Line~\ref{line:envy_cycle}.

Let $\tilde{X}_i$ denote the bundle held by agent $i$ at that point. Since agent $i$ does not envy agent $j'$ at the moment $g$ is added to $\tilde{X}_{j'}$, we have $v_i(\tilde{X}_i) \geq v_i(\tilde{X}_{j'})$. Moreover, by the design of the algorithm, the value each agent has for her own bundle can only increase over time. In particular, $v_i(X_i) \geq v_i(\tilde{X}_i) \geq v_i(\tilde{X}_{j'}) = v_i(X_j \setminus \{g\})$.

{We distinguish between two further cases.}
{If $v_i(g) \le \tfrac{1}{2}v_i(X_i)$, then}
\begin{align*}
\mu_i(2, X_i \cup X_j) &\le \tfrac{1}{2} (v_i(X_i) + v_i(X_j)) && (\text{since $\mu_i(2,S) \leq \tfrac{1}{2} v_i(S)$ for all $S$}) \\
&=  \tfrac{1}{2} (v_i(X_i) + v_i(X_j \setminus \{g\}) + v_i(g)) && (\text{by additivity of $v_i$}) \\  
&\le \tfrac{1}{2}(v_i(X_i) + v_i(X_i) + \tfrac{1}{2}v_i(X_i)) && (\text{by the  assumption and observation above}) \\
&= \tfrac{5}{4} v_i(X_i),    
\end{align*}
{implying Equation~\eqref{eq:4-5-bound}, as desired}.

Consider next the case where $v_i(g) > \tfrac{1}{2}v_i(X_i)$. Note that at the moment $g$ is added to $X_j$, the bundle $X_i$ may not yet have been fully formed; however, there must be some agent $i'$ holding a bundle $\tilde{X}_{i'} \subseteq X_i$ at that point, with all items in $X_i \setminus \tilde{X}_{i'}$ still unallocated. Since items are allocated in descending order of value, for every $g' \in \tilde{X}_{i'} \cup (X_j \setminus \{g\})$ we have
\begin{equation}\label{largeitems}
    v_i(g') \ge v_i(g) > \tfrac{1}{2} v_i(X_i).
\end{equation}
        
We now show that $|\tilde{X}_{i'}| = 1$. On one hand, since $v_i(\tilde{X}_{i'}) \leq v_i(X_i)$ and each item in $\tilde{X}_{i'}$ is worth at least $\tfrac{1}{2}v_i(X_i)$ to agent $i$ by Equation~\eqref{largeitems}, we have $|\tilde{X}_{i'}| \le 1$. On the other hand, since item $g$ is assigned to agent $j'$, agent $i'$ does not envy the bundle $\tilde{X}_{j'} = X_j \setminus \{g\}$; hence, by the non-degeneracy of the valuations, agent $i'$ must hold at least one item.

We next show that $|X_j| = 2$. By assumption, we already have $|X_j| \geq 2$, and strict inequality is impossible since $v_i(X_i) \geq v_i(X_j \setminus \{g\})$ and, by Equation~\eqref{largeitems}, each item in $X_j \setminus \{g\}$ is worth at least $\tfrac{1}{2}v_i(X_i)$ to agent $i$.

Denote $X_j = \{g', g\}$ and $\tilde{X}_{i'} = \{h\}$.
We next show that 
\begin{equation}
    \label{eq:ranking}
    v_i(h) > v_i(g') > v_i(g).
\end{equation}
Note that at the moment item $g$ is first allocated, we have $\tilde{X}_{j'} = \{g'\}$ and $\tilde{X}_{i'} = \{h\}$. Since items are allocated in descending order and, by assumption, no two items share the same value, it follows that $v_i(g') > v_i(g)$ and $v_i(h) > v_i(g')$. Moreover, since item $g$ is assigned to agent $j'$, agent $i'$ does not envy agent $j'$ at that point, so $v_{i'}(h) > v_{i'}(g')$; by the identical ranking assumption, this implies $v_i(h) > v_i(g')$. We get
\begin{align*}
            \mu_i(2, X_i \cup X_j) 
            &\le \mu_i(2, (X_i \cup X_j) \setminus (X_i \setminus \{h\})) + v_i(X_i \setminus \{h\}) && (\text{by Lemma~\ref{lemmamms}})\\ 
            &= \mu_i(2, \{g, g', h\}) + v_i(X_i \setminus \{h\}) && (\text{by the above}) \\
            &\leq v_i(h) + v_i(X_i \setminus \{h\}) && (\text{by Equation~\eqref{eq:ranking}})\\
            &= v_i(X_i) && (\text{by additivity})
        \end{align*}
{implying Equation~\eqref{eq:4-5-bound}, as desired}.
\end{proof}

\begin{remark}
In fact, our proof of the existence of a $4/5$-PMMS allocation holds
under the slightly weaker assumption that the agents agree on the
identity and ranking of the top $n$ items, i.e., there exist distinct items
$g_1, \ldots, g_n$ such that $v_i(g_1) \geq v_i(g_2) \geq \cdots \geq v_i(g_n)
\geq v_i(g)$ for every agent $i$ and every $g \notin \{g_1, \ldots,
g_n\}$. We note, however, that EFX requires the
identical-rankings assumption.
\end{remark}
\begin{remark}
In Section~\ref{sec:tightness_envy_cycles}, we show that this analysis is tight by exhibiting an instance on which Algorithm~\ref{alg:envy_cycles} achieves no better than a $4/5$ approximation of PMMS.
\end{remark}

\section{EGMMS for Bivalued Valuations}

In this section we establish the existence of EGMMS for instances with bivalued valuations.

\begin{definition}[Bivalued Valuations]
    An instance  has {bivalued valuations} if the valuations are additive and there exist common values $0 \leq a <  b$ 
   with $v_i(\{g\}) \in \{a, b\}$ for all  $i \in N$ and $g \in M$.
\end{definition}

\thmbivalued*

By~Proposition~\ref{prop:EPMMS_EGMMS_identical_rankings_WLOG}, to establish the existence of EGMMS for general additive valuations, it suffices to prove it for instances with identical rankings.
In this setting, given the common item ordering $M = (g_1, \dots, g_m)$, each agent $i$ has an index $k_i \in \{0, 1, \dots, m\}$ such that $v_i(g_j) = b$ for all $j \le k_i$, and $v_i(g_j) = a$ for all $j > k_i$. We define agent $i$'s \emph{$b$-valued prefix}, denoted by $p_b(i)$, as this set of high-valued items: $p_b(i) := \{ g_j \in M \mid 1 \le j \le k_i \}.$

Algorithm~\ref{alg:leximin-feige} builds upon the algorithmic framework of \citet*{Feige2022MMSab}, which computes MMS allocations for bivalued instances. The key difference lies in Line~\ref{line:leximin}: \citet*{Feige2022MMSab} select an arbitrary MMS partition according to agent~$i$, which suffices to establish the existence of MMS allocations for bivalued valuations. To obtain the stronger guarantee of EGMMS, however, an arbitrary MMS partition does not suffice, as we show in Observation~\ref{MMS_and_EPMMS_incomparable}. For this reason, we instead select the leximin partition.

\begin{algorithm}[t]
\caption{\citet*{Feige2022MMSab}'s algorithm with leximin partitions}
\label{alg:leximin-feige}
\KwIn{An instance with bivalued valuations and identical rankings.}
\KwOut{An EGMMS allocation.}
$M^{(1)} \gets M$;  Set indices s.t. $v_1(M) \ge \cdots \ge v_n(M)$ and 
$v_i(g_1) \ge \cdots \ge v_i(g_m)$ $\forall i$\;
\For{$i = 1, \dots, n$}{
    Let $(B_i^i, \dots, B_n^i)$ be the leximin partition of $M^{(i)}$ under $v_i$ into $(n{-}i{+}1)$ bundles\label{line:leximin}\;
    Rearrange bundles so that $B_i^i$ has minimum cardinality\label{line:min_cardinality}\;
    Rearrange items so that $B_i^i$ contains the $|B_i^i \cap p_b(i)|$ highest-indexed items of $p_b(i) \cap M^{(i)}$\;
    $X_i \leftarrow B_i^i$ and $M^{(i+1)} \gets M^{(i)} \setminus X_i$\;
}
\Return $X = (X_1, \dots, X_n)$\;
\end{algorithm}

Before presenting the full proof, we outline the high-level structure of the proof. 
The proof constructs, for an arbitrary agent $i$, a specific valid certificate $Y^i$ that combines the actual allocation for agents $1, 2, \ldots, i-1, i$ with agent $i$'s optimal leximin partition for the remaining items, thereby showing that the algorithm yields an EGMMS allocation.

To verify that this certificate satisfies the GMMS condition for every subset of agents $J$ containing $i$, we proceed by backward induction on the smallest index $j_1$ in $J$. The base case ($j_1 = i$) follows from the fact that leximin partitions satisfy GMMS.

For the inductive step ($j_1 < i$), the argument splits according to whether agent $j_1$ receives any of agent $i$'s high-value items. If so, the nested prefix structure forces agents $j_1$ through $i$ to agree on the values of all remaining items, and the leximin substitution lemma (Lemma~\ref{lem:leximin_substitution}) allows us to show the certificate suffix forms a valid leximin partition; GMMS then follows from valuation dominance.

Otherwise, agent $j_1$'s bundle contains only low-value items, and we split further by the size of that bundle: if it is of below-average size, the MMS monotonicity lemma (Lemma~\ref{lem:mms_monotonicity}) allows us to drop it and apply the inductive hypothesis to $J \setminus {j_1}$; if it is of above-average size, an averaging argument identifies an agent $j' \notin J$ with a strictly larger bundle, and swapping $j_1$ for $j'$ 
allows us to apply the inductive hypothesis since the minimum index in the subset increases,
while replacing 
agent $j_1$'s bundle
with 
agent $j'$'s 
larger bundle 
in any MMS partition can only increase the MMS value.

We next present the full proof of Theorem~\ref{thm:GMMS_bivalued}.
We first establish
two structural lemmas regarding optimal leximin and MMS partitions.
The first establishes a
substitution property for the leximin partition.

\begin{restatable}[Leximin Substitution]{lemma}{lemmaleximinsubstitution}
\label{lem:leximin_substitution}
Let $S$ be a set of items, $k \ge 2$ be an integer, and $v$ be an additive valuation function. Let $P = \{A_1, A_2, \dots, A_k\}$ be an optimal leximin partition of $S$ into $k$ bundles. If $P^\star$ is any optimal leximin partition of 
$S \setminus A_1$ into $k-1$ bundles under $v$, then $\{A_1\} \cup P^\star$ is also an optimal leximin partition of $S$ into $k$ bundles under $v$.
\end{restatable}
\begin{proof}
First, observe that the sub-partition $P \setminus \{A_1\} = \{A_2, \dots, A_k\}$ must itself be an optimal leximin partition of $S \setminus A_1$ into $k-1$ bundles. Otherwise, there would exist a partition of $S \setminus A_1$ with a strictly larger lexicographically sorted value vector, which, when combined with $A_1$, would yield a partition of $S$ that strictly dominates $P$, contradicting the optimality of $P$.

Since both $P^\star$ and $P \setminus \{A_1\}$ are optimal leximin partitions of the same instance, their sorted valuation vectors must be identical.

Now consider the partition $\{A_1\} \cup P^\star$. Its multiset of bundle values consists of $v(A_1)$ together with the values induced by $P^\star$. Because the sorted valuation vector of $P^\star$ matches that of $P \setminus \{A_1\}$, it follows that the sorted valuation vector of $\{A_1\} \cup P^\star$ is identical to that of $P$.

Since $P$ is an optimal leximin partition of $S$, any partition with the same lexicographically sorted valuation vector is also optimal. Therefore, $\{A_1\} \cup P^\star$ is a  leximin partition of $S$.
\end{proof}

We next establish the following monotonicity property of the maximin partition.

\begin{restatable}[MMS Monotonicity]{lemma}{lemmammsmonotonicity}
\label{lem:mms_monotonicity}
Let $M$ be a set of $m$ items, $n \ge 2$ be an integer, and $v$ be an additive valuation function. Let $X \subseteq M$ be a subset with $|X| \le {m}/{n}$, consisting of the least-valued items in $M$. 
Then, $\mu(n-1, M \setminus X) \ge \mu(n, M)$.
\end{restatable}
\begin{proof}
Let $P = \{A_1, \dots, A_n\}$ be an MMS partition of $M$ into $n$ bundles, i.e., $v(A_j) \ge \mu(n, M)$ for all $j \in [n]$.
By the pigeonhole principle, since $m$ items are distributed among $n$ bundles, there exists a bundle $B \in P$ such that $|B| \ge {m}/{n}$.

Since $P \setminus \{B\}$ is a valid partition of $M \setminus B$ into $n-1$ bundles, and for every $A \in P \setminus \{B\}$ we have $v(A) \ge \mu(n, M)$, it must hold that
\begin{equation}
\label{mms_removing_set_from_partition}
\mu(n-1, M \setminus B) \ge \mu(n, M).
\end{equation}

In addition, since $|X| \le {m}/{n} \le |B|$, it follows that $|X| \le |B|$. Moreover, because $X$ consists of the least-valued items in $M$, removing $X$ can only increase (or at least not decrease) the MMS compared to removing $B$. Therefore,
\begin{equation}
\label{mms_removing_smaller_bundle}
\mu(n-1, M \setminus X) \ge \mu(n-1, M \setminus B).
\end{equation}

Combining~\eqref{mms_removing_set_from_partition} and~\eqref{mms_removing_smaller_bundle} yields
\[
\mu(n-1, M \setminus X) \ge \mu(n, M). \qedhere
\]
\end{proof}

We are now ready to prove the main theorem of this section.

\begin{proof}[Proof of Theorem~\ref{thm:GMMS_bivalued}] 
 We claim that Algorithm~\ref{alg:leximin-feige} produces an EGMMS allocation. 
Since the algorithm orders the agents by their total valuation of the grand bundle ($v_1(M) \ge v_2(M) \ge \dots \ge v_n(M)$), and valuations are bivalued with identical rankings, this ordering coincides with the lengths of their $b$-valued prefixes. Thus, $|p_b(1)| \ge |p_b(2)| \ge \dots \ge |p_b(n)|$, and in particular, for any $j < i$ we have $p_b(i) \subseteq p_b(j)$.
We now show that the resulting allocation $X$ is EGMMS. 
Fix an agent $i$. We construct a valid EGMMS certificate $Y^i$ for agent $i$:
\[
Y^i = (X_1, X_2, \ldots, X_i, {B}^i_{i+1}, B^i_{i+2}, \ldots, {B}^i_n).
\]
Namely, $Y^i_j = X_j$ for $j \le i$, and $Y^i_j = {B}^i_j$ for $j > i$. 
We claim that this partition forms an EGMMS certificate for agent $i$.
Since it forms a partition of the entire set of items $M$, and $Y^i_i = X_i$, it remains to show that agent $i$ is GMMS-satisfied with $Y^i$. That is, we need to show that for every subset $J \subseteq [n]$ with $i \in J$, it holds that:

\begin{equation}
\label{eq:gmms-cert}
v_i(X_i) \ge \mu_i(|J|, \mathcal{M}_J), \quad \mbox{where} \quad \mathcal{M}_J = \bigcup_{j \in J} Y^i_j.
\end{equation}

Let $j_1 = \min(J)$. Because $i \in J$, it holds that $j_1 \le i$. 
We prove Equation~\eqref{eq:gmms-cert} by (backward) induction on the minimum index $j_1$, starting from $j_1 = i$ down to $j_1 = 1$.

\textbf{Base Case ($j_1 = i$):} 
If $\min(J) = i$, then $J \subseteq \{i, i+1, \dots, n\}$. The subset $\{Y^i_j \mid j\in J\} = \{B^i_j \mid j\in J\} $ consists of bundles corresponding to a subset of agent $i$'s leximin partition of $M^{(i)}$. By Proposition~\ref{prop:leximin_is_gmms}, this is a GMMS allocation for the instance in which all agents have identical valuations $v_i$. Thus $v_i(X_i) \ge \mu_i(|J|, \mathcal{M}_J)$, as desired.

\textbf{Inductive Step ($j_1 < i$):} 
Assume the GMMS condition holds for all valid subsets $J'$ where $i \in J'$ and $\min(J') > j_1$, i.e., $v_i(X_i) \ge \mu_i(|J'|, \bigcup_{j \in J'} Y^i_j)$.

\textbf{Case 1: $Y^i_{j_1} \cap p_b(i)  \ne \emptyset$} (some item in $p_b(i)$ is allocated to agent $j_1$).
Since $j_1 < i$, it holds that $p_b(i) \subseteq p_b(j_1)$. In the algorithm, agent $j_1$ takes the 
\emph{highest-indexed}
available items from $p_b(j_1)$. If agent $j_1$ is forced to take an item from $p_b(i)$ (which 
is a prefix of
$p_b(j_1)$), it implies that all available items in the disagreement zone, $p_b(j_1) \setminus p_b(i)$, 
have already been allocated either to agent $j_1$ or to an earlier agent.

Consequently, for every remaining item $t \in M^{(j_1 + 1)}$, the valuations of agents $j_1$ and $i$ are identical. Since agent $i$ has the shortest prefix among all agents between $j_1$ and $i$
{by the assumption on the ordering of agents}, 
it holds that for every agent $\ell \in \{j_1, \dots, i\}$, 
{we have $p_b(i) \subseteq p_b(\ell) \subseteq p_b(j_1)$, and thus,}
for every remaining item $t \in M^{(j_1 + 1)}$, the valuations are identical: $v_{\ell}(t) = v_{j_1}(t)$. 

We now prove that the certificate sequence $(Y^i_{j_1}, Y^i_{j_1+1}, \dots, Y^i_{n})$ forms a valid leximin partition of $M^{(j_1)}$ under valuation $v_{j_1}$. We show this by working backward from agent $i$ and applying Lemma~\ref{lem:leximin_substitution}.

First, consider the items $M^{(i)}$ remaining at agent $i$'s turn, which are partitioned into the suffix $(Y^i_i, \dots, Y^i_n)$. By definition of the certificate, this sequence corresponds to an optimal leximin partition under $v_i$. Since $v_i$ and $v_{j_1}$ are identical over these items, $(Y^i_i, \dots, Y^i_n)$ is also an optimal leximin partition of $M^{(i)}$ under $v_{j_1}$.

Now, assume inductively that for some index $k \in \{j_1 + 1, \dots, i\}$, the suffix $(Y^i_{k}, \dots, Y^i_n)$ is an optimal leximin partition of $M^{(k)}$ under $v_{j_1}$. At step $k-1$, agent $k-1$ computes an optimal leximin partition of $M^{(k-1)}$ under $v_{k-1}$ 
{and is assigned bundle $Y^i_{k-1} = X_{k-1}$.}
{Observe that $v_{k-1}$ and $v_{j_1}$ agree on $M^{(k-1)}$: this follows from the observation above when $k-1 > j_1$, and trivially when $k-1 = j_1$. Thus,}
 bundle $Y^i_{k-1}$ is also 
{part of a valid}
leximin partition of $M^{(k-1)}$ under $v_{j_1}$.

By Lemma~\ref{lem:leximin_substitution}, since $Y^i_{k-1}$ is a bundle from 
{a valid}
leximin partition under $v_{j_1}$, and the subsequent sequence $(Y^i_{k}, \dots, Y^i_n)$ forms 
{a valid}
leximin partition of the remaining items $M^{(k)} = M^{(k-1)} \setminus Y^i_{k-1}$ under $v_{j_1}$, joining them yields an optimal leximin partition. Thus, $(Y^i_{k-1}, Y^i_k, \dots, Y^i_n)$ is 
{a valid}
leximin partition of $M^{(k-1)}$ under $v_{j_1}$.

Applying this logic inductively backward from $i$ down to $j_1$, we conclude that the entire constructed sequence $(Y^i_{j_1}, \dots, Y^i_n)$ forms 
{a valid}
leximin partition of $M^{(j_1)}$ under valuation $v_{j_1}$. 

Because a leximin partition is a GMMS allocation for identical agents (Proposition~\ref{prop:leximin_is_gmms}),  it holds that:
\begin{equation}
\label{eq:case1_bundle_value}
    v_{j_1}(X_i) = v_{j_1}(Y^i_i) \ge \mu_{j_1}(|J|, \mathcal{M}_J).
\end{equation}

Moreover, since $X_i \subseteq M^{(j_1 + 1)}$, agent $i$ and agent $j_1$ value its items identically:
\begin{equation}
\label{eq:case1_identical_val}
    v_i(X_i) = v_{j_1}(X_i).
\end{equation}

Finally, because $j_1 < i$, agent $j_1$ has at least as many $b$-valued items in $\mathcal{M}_J$ as agent $i$ does. Consequently, agent $j_1$'s valuation weakly dominates agent $i$'s. Thus, the maximin share for agent $j_1$ is at least as large as that for agent $i$:
\begin{equation}
\label{eq:case1_mms_domination}
    \mu_{j_1}(|J|, \mathcal{M}_J) \ge \mu_i(|J|, \mathcal{M}_J).
\end{equation}

Chaining \eqref{eq:case1_identical_val}, \eqref{eq:case1_bundle_value}, and \eqref{eq:case1_mms_domination} together, we conclude that
\[
v_i(X_i) \ge \mu_i(|J|, \mathcal{M}_J).
\]

\textbf{Case 2: $ Y^i_{j_1}  \cap p_b(i) = \emptyset$} (No item of $p_b(i)$ is allocated to agent $j_1$). 
In this case, according to $v_i$, all items in $X_{j_1}$ have value $a$. 

We distinguish between two further cases. Consider first the case where $|X_{j_1}| \le {|\mathcal{M}_J|}/{|J|}$. \\
    Let $J' = J \setminus \{j_1\}$. Because $j_1$ 
    {is}
    the unique minimum index of $J$, we have $\min(J') > j_1$. By the inductive hypothesis, we have $v_i(X_i) \ge \mu_i(|J'|, \mathcal{M}_{J'})$.

    By Lemma~\ref{lem:mms_monotonicity}, since $|X_{j_1}| \le {|\mathcal{M}_J|}/{|J|}$ and $X_{j_1}$ consists only of items of minimum value $a$ (under $v_i$), it follows that
\[
\mu_i(|J'|, \mathcal{M}_{J'}) = \mu_i(|J|-1, \mathcal{M}_J \setminus X_{j_1}) \ge \mu_i(|J|, \mathcal{M}_J).
\]

    Chaining this with the inductive hypothesis yields $v_i(X_i) \ge \mu_i(|J|, \mathcal{M}_J)$, completing the proof for this subcase.
    
    Now, let us examine the complementary case where $|X_{j_1}| > {|\mathcal{M}_J|}/{|J|}$. 
    In Algorithm~\ref{alg:leximin-feige}, $X_{j_1}$ is chosen as the minimum cardinality bundle of a partition of $M^{(j_1)}$ into $n-j_1+1$ bundles. Thus, $|X_{j_1}| \le {|M^{(j_1)}|}/({n-j_1+1})$. 

    Combining this with our current assumption, we have:
    \begin{align}
        \frac{|M^{(j_1)}|}{n-j_1+1} > \frac{|\mathcal{M}_J|}{|J|}. \label{ineq:aver1}
    \end{align}
    Denote $N_{j_1} = \{j_1, j_1 + 1 , \dots, n \}$ as the suffix of agents from $j_1$ to $n$. Notice that the bundles $\{Y^i_k \mid k \in N_{j_1}\}$ form a partition of the items in $M^{(j_1)}$. The left-hand side of 
    Inequality~\ref{ineq:aver1}
    is the average bundle size across this entire partition, while the right-hand side
    of Inequality~\ref{ineq:aver1} is the average bundle size of the subset of bundles indexed by $J \subseteq N_{j_1}$. 

    Because the average size of the subset $J$ is strictly less than the overall average, the average size of the remaining bundles indexed by $N_{j_1} \setminus J$ must be strictly greater than the overall average. Therefore, there must exist at least one bundle in the complement set, indexed by some $j' \in N_{j_1} \setminus J$, whose size is strictly greater than the overall average:
    \begin{equation}
    \label{case2.2_eq1}
            |Y^i_{j'}| > \frac{|M^{(j_1)}|}{n-j_1+1} \ge |X_{j_1}|.
    \end{equation}

    \begin{figure}[t]
\centering
\begin{tikzpicture}[scale=1, every node/.style={font=\normalsize}]
    \def\w{0.9}
    \def\h{0.7}

    \foreach \i in {1,...,14} {
        \draw (\i*\w,0) rectangle ++(\w,\h);
    }

    \foreach \i in {6,7,9,10,13} {
        \fill[cyan!40, fill opacity=0.5] (\i*\w,0) rectangle ++(\w,\h);
    }

    \node at (1.5*\w,0.35) {$Y^i_1$};
    \node at (2.5*\w,0.35) {$Y^i_2$};
    \node at (3.5*\w,0.35) {$\dots$};
    \node at (6.5*\w,0.35) {$Y^i_{j_1}$};
    \node at (7.5*\w,0.35) {$\dots$};
     \node at (9.5*\w,0.35) {$Y^i_{i}$};
    \node at (11.5*\w,0.35) {$Y^i_{j'}$};
    \node at (12.5*\w,0.35) {$\dots$};
    \node at (14.5*\w,0.35) {$Y^i_n$};

    \draw[<->] (6.5*\w,0.8) .. controls +(0,0.6) and +(0,0.6) .. (11.5*\w,0.8);

    \node at (9*\w,1.6) { 
    $J' = J \setminus \{j_1\} \cup \{j'\}$
    };
\end{tikzpicture}
\caption{Illustration of the swapping argument used in the proof of Theorem~\ref{thm:GMMS_bivalued}. 
The  bundles highlighted in blue correspond to the set $J$. 
Note that it might be that $j' < i$.
} 
\label{fig:swap-argument}
\end{figure}
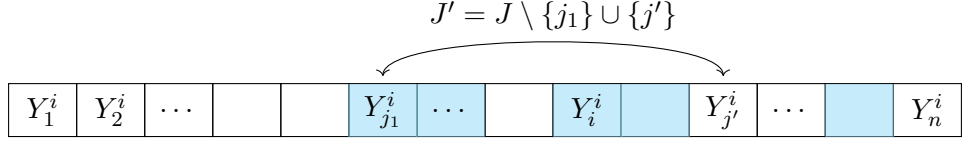

    Let $J' = (J \setminus \{j_1\}) \cup \{j'\}$; see Figure~\ref{fig:swap-argument} for an illustration.
    Because $j_1$ is the unique minimum element of $J$, and $j' \in N_{j_1} \setminus J$ implies $j' > j_1$, it holds that $\min(J') > j_1$. 
    By the inductive hypothesis, we apply the GMMS condition to $J'$:
  \begin{equation}
    \label{case2.2_eq2}
    v_i(X_i) \ge \mu_i(|J'|, \mathcal{M}_{J'}).
\end{equation}
    
    We now compare the maximin share for agent $i$ of $\mathcal{M}_{J'}$ to that of $\mathcal{M}_J$. Notice that $\mathcal{M}_{J'} = (\mathcal{M}_J \setminus X_{j_1}) \cup Y^i_{j'}$. 
    From the condition of Case 2, every item in $X_{j_1}$ has a value of exactly $a$ according to $v_i$. Meanwhile, every item in $Y^i_{j'}$ has a value of at least $a$. Also, $|Y^i_{j'}| > |X_{j_1}|$.
    
    Consider the MMS partition of $\mathcal{M}_J$ into $|J|$ bundles based on $v_i$. If we replace the items of $X_{j_1}$ in this partition with exactly $|X_{j_1}|$ items from $Y^i_{j'}$, and distribute the remaining items of $Y^i_{j'}$ arbitrarily, the value of every bundle in the partition weakly increases. This forms a valid partition of $\mathcal{M}_{J'}$ into $|J'|$ bundles where the minimum bundle is at least as valuable as the minimum bundle in the optimal partition of $\mathcal{M}_J$. Thus:
   \begin{equation}
    \label{case2.2_eq3}
    \mu_i(|J'|, \mathcal{M}_{J'}) \ge \mu_i(|J|, \mathcal{M}_J).
\end{equation}
    
    Chaining \eqref{case2.2_eq2} and \eqref{case2.2_eq3} yields:
    \[
    v_i(X_i)  \ge \mu_i(|J|, \mathcal{M}_J).
    \]
    This satisfies the GMMS condition for this subcase, completing the inductive step.

Therefore, agent $i$ is GMMS-satisfied with $Y^i$ for all valid subsets $J$. 
Since this holds for any arbitrary agent $i$, the allocation $X$ is EGMMS for the ordered instance. By Proposition~\ref{prop:EPMMS_EGMMS_identical_rankings_WLOG}, mapping this allocation back to the original items yields an EGMMS allocation for the general instance, completing the proof.
\end{proof}

Finally, even though computing leximin allocations is NP-hard for general additive valuations \cite{PlautR17}, Algorithm~\ref{alg:leximin-feige} can be implemented in polynomial time, as a consequence of the following observation.
    
\begin{restatable}{observation}{obsleximinbivaluedpoly}
\label{obs:leximin_bivalued_poly}
There exists a polynomial time algorithm that computes the leximin allocation for any instance with identical bivalued valuations.
\end{restatable}

\begin{proof}
Let $n$ be the number of agents, $N_a$ the number of items of value $a$, and $N_b$ the number of items of value $b$. Consider the following dynamic programming solution. For all $1 \leq k \leq n$, $0 \leq n_a \leq N_a$, and $0 \leq n_b \leq N_b$, let
\[
\textsc{DP}[k][n_a][n_b]
\]
denote the sorted (non-decreasing) value vector of a leximin allocation of the sub-instance consisting of $k$ agents with identical bivalued valuations, $n_a$ items of value $a$, and $n_b$ items of value $b$.

For the base case $k = 1$, the unique allocation gives the single agent every item, so we set $\textsc{DP}[1][n_a][n_b] = (n_a \cdot a + n_b \cdot b)$.

For $k > 1$, assume inductively that $\textsc{DP}[k'][n_a'][n_b']$ has already been computed for every triple $(k', n_a', n_b')$ with $k' \le k$, $n_a' \le n_a$, $n_b' \le n_b$, and at least one of these inequalities strict. We iterate over all pairs $(i, j)$ with $0 \le i \le n_a$ and $0 \le j \le n_b$, where $(i, j)$ corresponds to assigning a bundle of $i$ items of value $a$ and $j$ items of value $b$ to one agent. For each such pair, we form a candidate vector $c[i][j]$ by inserting the value $i \cdot a + j \cdot b$ into the vector $\textsc{DP}[k-1][n_a - i][n_b - j]$ and re-sorting in non-decreasing order. We then set $\textsc{DP}[k][n_a][n_b]$ to be the lexicographically maximal vector among all such candidates $c[i][j]$.

Since the number of states is $n N_a N_b$ and the number of pairs iterated over in each state is at most $N_a N_b$, the algorithm runs in polynomial time.\end{proof}

\section{EPMMS for Three Agents and for Two Types of Agents}
In this section, we establish the existence of EPMMS allocations in two important restricted settings.
We begin with the three-agent case.

\thmthreeagents*

\begin{proof}
Consider the case of three agents, and suppose that $v_2$ and $v_3$ are PMMS-feasible.
Let $(X_1, X_2, X_3)$ be agent~1's leximin partition.
Assume without loss of generality that $X_1$ is the least preferred bundle according to agent~2, i.e., $v_2(X_2) \ge v_2(X_1)$ and $v_2(X_3) \ge v_2(X_1)$.

Let $(X_2', X_3')$ be agent $2$'s leximin partition of
 $X_2 \cup X_3$ into two bundles. 
We then allocate to agent $3$ her most preferred bundle among $X_1, X_2', X_3'$.

We next describe the allocation of agents $1$ and $2$. If $X_1$ is still available, then agent~1 takes $X_1$, and agent~2 receives the remaining bundle. Otherwise, if $X_1$ is taken by agent $3$, then agent~2 chooses her preferred bundle among $X_2, X_3$ (the bundles from agent~1's partition), and agent~1 receives the remaining bundle. See Figures~\ref{fig:three_agents_1} and~\ref{fig:three_agents_2} for an illustration of the allocation.

\begin{figure}[t]
    \centering
    \begin{subfigure}[t]{0.48\textwidth}
        \centering
        \begin{tikzpicture}[every node/.style={draw,circle,minimum size=0.65cm,inner sep=1pt},
                            box/.style={draw,red,thick,rectangle,rounded corners,inner sep=4pt}]
        \node (x1) at (0,1.1) {$1$};
        \node (x2) at (1.4,1.1) {};
        \node (x3) at (2.8,1.1) {};
        \node[draw=none] at (0,0.35) {$X_1$};
        \node[draw=none] at (1.4,0.35) {$X_2$};
        \node[draw=none] at (2.8,0.35) {$X_3$};
        \node (x2p) at (1.4,-0.55) {$2$};
        \node (x3p) at (2.8,-0.55) {$3$};
        \node[draw=none] at (1.4,-1.35) {$X_2'$};
        \node[draw=none] at (2.8,-1.35) {$X_3'$};
        \node[box,fit=(x1)] {};
        \node[box,fit=(x2p)(x3p)] {};
        \end{tikzpicture}
        \caption{Case 1: $X_1$ remains available. Agent 3 selects from $\{X_1,X_2',X_3'\}$; agent 1 receives $X_1$, and agent 2 receives the remaining bundle among $X_2'$, $X_3'$.}
        \label{fig:three_agents_1}
    \end{subfigure}
    \hfill
    \begin{subfigure}[t]{0.48\textwidth}
        \centering
        \begin{tikzpicture}[
        circ/.style={draw,circle,minimum size=0.65cm,inner sep=1pt},
        box/.style={draw,red,thick,rectangle,rounded corners,inner sep=5pt}
        ]
        \node[circ] (a3) at (0,1.1) {$3$};
        \node[circ] (a1) at (1.4,1.1) {$1$};
        \node[circ] (a2) at (2.8,1.1) {$2$};
        \node at (0,0.3) {$X_1$};
        \node at (1.4,0.3) {$X_2$};
        \node at (2.8,0.3) {$X_3$};
        \node[circ] (x2p) at (1.4,-0.55) {};
        \node[circ] (x3p) at (2.8,-0.55) {};
        \node at (1.4,-1.2) {$X_2'$};
        \node at (2.8,-1.2) {$X_3'$};
        \node[box,fit=(a3)(a1)(a2)] {};
        \end{tikzpicture}
        \caption{Case 2: Agent 3 takes $X_1$. Agent 2 then chooses her preferred bundle among $X_2$ and $X_3$, and agent 1 receives the remaining bundle.}
        \label{fig:three_agents_2}
    \end{subfigure}
    \caption{The two cases in the proof of Theorem~\ref{thm:threeagents}.}
    \label{fig:three_agents}
\end{figure}

We first analyze agent $3$'s allocation. Agent~3 receives the maximal bundle in the partition $(X_1, X_2', X_3')$. Thus, the partition $(X_1, X_2', X_3')$ serves as her EPMMS certificate: by Proposition~\ref{prop:ef_implies_pmms}, since she receives the maximal bundle in this partition, she does not PMMS-envy any other bundle in the partition since $v_3$ is PMMS-feasible. Moreover, if we assume that her valuation is also GMMS-feasible, then she receives at least her MMS value.

Next, we argue that agent~2 does not PMMS-envy any other agent.
    
    If she receives one of the bundles in $(X_2', X_3')$, then, by the choice of $(X_2', X_3')$,
    \[
    \min\{v_2(X_2'), v_2(X_3')\}
    \ge
    \min\{v_2(X_2), v_2(X_3)\}
    \ge v_2(X_1).
    \]
    Thus, she does not PMMS-envy the bundle $X_1$. Moreover,
    \[
    \min\{v_2(X_2'), v_2(X_3')\} \ge \mu_2(2, X_2' \cup X_3'),
    \]
    and therefore she does not PMMS-envy the remaining bundle among $X_2'$ and $X_3'$.
    
    If she receives her most preferred bundle in $(X_2, X_3)$, then she receives the maximal bundle in the allocation; thus, she does not PMMS-envy any other agent by Proposition~\ref{prop:ef_implies_pmms} since $v_2$ is PMMS-feasible. Since the final allocation is a PMMS allocation from the perspective of agent~2, it also serves as her EPMMS certificate.

    Finally, since agent~1 receives one of $X_1, X_2, X_3$, The partition $(X_1, X_2, X_3)$ serves as her EPMMS certificate, since it is agent~1 leximin partition. Moreover, 
    since the leximin partition is also an MMS partition, agent~1 receives at least her MMS share. 
\end{proof}

We next consider the case of two agent types, for which we prove the following theorem.

\thmtwotypes*

\begin{proof}
We consider the case where there are only two types of valuation functions, $v_{\alpha}$ and $v_{\beta}$. That is, for each agent $i$, either $v_i = v_{\alpha}$ or $v_i = v_{\beta}$. Let $N_{\alpha}$ be the set of agents whose valuation functions are $v_{\alpha}$, and let $N_{\beta}$ be the set of agents whose valuation functions are $v_{\beta}$. Assume that $v_\beta$ is an \mmsf valuation. 

Let $(X_1, X_2, \ldots, X_n)$ be a leximin partition of $M$ into $n$ bundles according to $v_{\alpha}$. Assume without loss of generality that $X_1, \ldots, X_{|N_{\alpha}|}$ are the $|N_{\alpha}|$ least-preferred bundles according to $v_{\beta}$. 

Let $(X_{|N_{\alpha}| + 1}', \ldots, X_n')$ be a leximin partition of $\bigcup_{i = |N_{\alpha}| + 1}^{n} X_i$ into $|N_{\beta}|$ bundles according to the valuation $v_{\beta}$.

We now argue that the allocation $(X_1, \ldots, X_{|N_{\alpha}|}, X_{|N_{\alpha}| + 1}', \ldots, X_n')$, in which the agents in $N_{\alpha}$ receive bundles  $X_1, \ldots, X_{|N_{\alpha}|}$ and the agents in $N_{\beta}$ receive bundles $X_{|N_{\alpha}| + 1}', \ldots, X_n'$, is an EPMMS allocation.

First, we show that agents in $N_{\beta}$ satisfy the PMMS condition. By the assumption, 
$X_1,\ldots,X_{|N_{\alpha}|}$ are the $|N_{\alpha}|$ bundles of minimum value according to $v_{\beta}$. 
Hence every bundle in $(X_{|N_{\alpha}|+1},\ldots,X_n)$ has value at least
\[
\max_{j \le |N_{\alpha}|} v_{\beta}(X_j).
\]
Now consider the leximin partition $(X_{|N_{\alpha}|+1}',\ldots,X_n')$ of 
$\bigcup_{i=|N_{\alpha}|+1}^{n} X_i$ with respect to $v_{\beta}$. By the definition of the leximin 
solution, the value of the minimum bundle in this partition is at least the minimum value in the 
original partition $(X_{|N_{\alpha}|+1},\ldots,X_n)$. Therefore, for every 
$i > |N_{\alpha}|$,
\[
v_{\beta}(X_i') \ge \min_{j > |N_{\alpha}|} v_{\beta}(X'_j) \geq \min_{j > |N_{\alpha}|} v_{\beta}(X_j)
\ge \max_{j \le |N_{\alpha}|} v_{\beta}(X_j).
\]
Thus each agent in $N_{\beta}$ prefers their bundle to any bundle allocated to an agent 
in $N_{\alpha}$, and since $v_\beta$ is an \mmsf valuation, they do not PMMS-envy agents in $N_{\alpha}$. 

Among the agents in $N_{\beta}$, the bundles 
$X_{|N_{\alpha}|+1}',\ldots,X_n'$ form a leximin partition under the identical valuation 
function $v_{\beta}$. Since leximin allocations with identical valuations are known to satisfy 
PMMS by Proposition~\ref{prop:leximin_is_gmms}, no pair of agents in $N_{\beta}$ exhibits PMMS envy.

Finally, we show that agents in $N_{\alpha}$ satisfy EPMMS. The partition $(X_1, X_2, \ldots, X_n)$ is an EPMMS certificate for the agents in $N_{\alpha}$, since it is a leximin partition under the identical valuation 
function $v_{\alpha}$, and leximin allocations with identical valuations are known to satisfy 
PMMS by Proposition~\ref{prop:leximin_is_gmms}.
\end{proof}

\section{Conclusions and Open Problems}

In this work, we initiate the study of epistemic PMMS allocations. We establish the existence of EPMMS in several important cases, along with a constant-factor approximation for general additive valuations. Whether exact EPMMS allocations exist for arbitrary additive valuations, even for four agents, 
is one of our major open problems.

{For the special case of bivalued valuations, we obtain the stronger guarantee of epistemic GMMS, which also implies the known existence result for MMS \cite{Feige2022MMSab}. Two important questions remain open for bivalued valuations: (a) whether PMMS allocations always exist (which would also imply the known existence result for EFX \cite{AmanatidisBFHV20,JinT25,ByrkaMP25}), and (b) the harder question of whether GMMS allocations always exist (which would imply both PMMS and EGMMS)}.

Many open problems arise from our separation results for fairness concepts (depicted in Figure~\ref{fig:implications}) and valuation classes (depicted in Figure~\ref{fig:gsandmmsfeasible}).
For example,  Proposition~\ref{prop:mms_and_pmms_does_not_imply_egmms} shows that under additive valuations, PMMS and MMS together do not imply EGMMS. An interesting open question is whether this implication fails also under bivalued valuations.

While we show that IMMP exists for three additive agents, an important direction for future research is determining whether IMMP or IMMX allocations always exist for general additive valuations. Note that IMMP existence would follow from PMMS existence, and IMMX existence would follow from EFX existence. However, even if PMMS (resp. EFX) does not always exist, IMMP (resp. IMMX) might still exist, and even if PMMS (resp. EFX) exists, the existence of IMMP (resp. IMMX) may be easier to prove.

\bibliographystyle{plainnat}
\bibliography{ref}

\appendix

\section{Separations between EPMMS and EEFX}\label{sec:separations}

\subsection{PMMS Is Not Preserved Under Envy-Cycle Elimination}\label{pmmsincycleelimination}

We begin by highlighting a fundamental difference between PMMS and EFX. The central operation of the Envy-Cycle Elimination algorithm (Algorithm~\ref{alg:envy_cycles}) is the envy-cycle elimination step (Line~\ref{line:envy_cycle}): we identify a cycle of agents in which each agent prefers the next agent's bundle, and we rotate the bundles along the cycle. It is well known that this step preserves the EFX guarantee, since the set of bundles is unchanged and each agent can only increase the value of her own bundle~\cite{PlautR17}.

In contrast, PMMS is not preserved under envy-cycle elimination. We construct an instance in which a PMMS allocation, after a single application of the procedure, becomes only $\tfrac{2}{3}$-PMMS.

\begin{proposition}
    For every $0 < \alpha \leq 1$, there exists an instance with additive valuations and an $\alpha$-PMMS allocation such that a single envy-cycle elimination step (i.e., one execution of Line~\ref{line:envy_cycle} in Algorithm~\ref{alg:envy_cycles}) produces an allocation that fails to be $\frac{2\alpha + \epsilon}{2+\alpha}$-PMMS for any $\epsilon > 0$.
\end{proposition}
\begin{proof}
    Consider an instance with $n = 3$ agents and six goods, where agent $1$ has the additive valuation
    \[
        v_1(g_1) = 2\alpha, \quad v_1(g_2) = 2, \quad v_1(g_3) = 2, \quad v_1(g_4) = \alpha, \quad v_1(g_5) = \alpha, \quad v_1(g_6) = \epsilon/2,
    \]
    and $\epsilon > 0$ is arbitrarily small. Consider the initial allocation $X$ given by
    \[
        X_1 = \{g_1\}, \quad X_2 = \{g_2, g_3\}, \quad X_3 = \{g_4, g_5, g_6\}.
    \]
    We first verify that $X$ satisfies the $\alpha$-PMMS condition for agent $1$, whose current value is $v_1(X_1) = 2\alpha$. Comparing her bundle against each of the other agents', we have
    \begin{align*}
        \alpha \cdot \mu_1(2, X_1 \cup X_2) &= \alpha \cdot \mu_1(2, \{g_1, g_2, g_3\}) = \alpha \cdot \min \{ 2,\, 2 + 2\alpha \} = 2\alpha = v_1(X_1), \\
        \mu_1(2, X_1 \cup X_3) &= \mu_1(2, \{g_1, g_4, g_5, g_6\}) = \min \{ 2\alpha,\, 2\alpha + \epsilon/2 \} = 2\alpha = v_1(X_1).
    \end{align*}
    Hence agent $1$'s $\alpha$-PMMS condition is satisfied.
    
    Now suppose an envy cycle is resolved (e.g., $1 \to 3 \to 1$), and as a result agent $1$ receives $X_3$ in place of $X_1$. Her value increases to $v_1(X_3) = 2\alpha + \epsilon/2$, and the maximin share again agent $2$ becomes
    \[
        \mu_1(2, X_2 \cup X_3) = \mu_1(2, \{g_2, g_3, g_4, g_5, g_6\}) = \min \{ 2 + \alpha,\, 2 + \alpha + \epsilon/2 \} = 2 + \alpha.
    \]
    The resulting ratio is
    \[
        \frac{v_1(X_3)}{\mu_1(2, X_2 \cup X_3)} = \frac{2\alpha + \epsilon/2}{2 + \alpha} < \frac{2\alpha + \epsilon}{2 + \alpha},
    \]
    so agent $1$'s bundle violates the $\frac{2\alpha + \epsilon}{2 + \alpha}$-PMMS guarantee after the envy-cycle elimination step.
\end{proof}

We next show that this bound is tight: eliminating an envy cycle in a PMMS allocation always results in an allocation that is $\tfrac{2}{3}$-PMMS.

\begin{proposition}\label{23pmms}
    Under additive valuations and for $0 < \alpha \leq 1$, a single envy-cycle elimination step (i.e., one execution of Line~\ref{line:envy_cycle} in Algorithm~\ref{alg:envy_cycles}) applied to an $\alpha$-PMMS allocation produces a $\frac{2\alpha}{2+\alpha}$-PMMS allocation.
\end{proposition}

\begin{proof}
Let $X = (X_1, \ldots, X_n)$ be the allocation immediately before Line~\ref{line:envy_cycle} of Algorithm~\ref{alg:envy_cycles} is executed, and let $X' = (X'_1, \ldots, X'_n)$ denote the allocation after its execution.
    Fix an agent $i$. 
    Since $X$ is $\alpha$-PMMS, we have $v_i(X_i) \geq \alpha \cdot \mu_i(2, X_i \cup X_j)$ for all $j$.
    The envy-cycle elimination step only permutes bundles among agents, so the set of bundles is unchanged; hence $v_i(X_i) \geq \alpha \cdot \mu_i(2, X_i \cup X'_j)$ for all $j$ as well.
    Our goal is to show that
    \begin{align}
        v_i(X'_i) \geq \frac{2\alpha}{2+\alpha} \cdot \mu_i(2, X'_i \cup X'_j) \quad \text{for all $j$}. \label{eq:goalcycles}
    \end{align}
    Fix an agent $j$, set $\mu' = \mu_i(2, X'_i \cup X'_j)$, and let $(A, B)$ be a partition of $X'_i \cup X'_j$ 
    such that $v_i(A) \geq \mu'$ and $v_i(B) \geq \mu'$. Since $v_i(X'_i) = v_i(A \cap X'_i) + v_i(B \cap X'_i)$, we may assume without loss of generality that $v_i(B \cap X'_i) \leq \tfrac{1}{2} v_i(X'_i)$. Observe that
    \begin{align*}
        (B \setminus X'_i) \cup X_i \cup (X'_j \setminus B) = X_i \cup X'_j,
    \end{align*}
    so $(B \setminus X'_i, X_i \cup (X'_j \setminus B))$ is a partition of $X_i \cup X'_j$. It follows that
    \begin{align*}
        \mu_i(2, X_i \cup X'_j) &\geq \min \bigl\{v_i(B \setminus X'_i),\; v_i(X_i \cup (X'_j \setminus B)) \bigr\} \\
        &= \min \bigl\{ v_i(B) - v_i(B \cap X'_i),\; v_i(X_i) + v_i(A) - v_i(A \cap X'_i)\bigr\}.
    \end{align*}
    We distinguish two cases according to which term attains this minimum.
    
    \textbf{Case 1}: $ \mu_i(2, X_i \cup X'_j) \geq v_i(B) - v_i(B \cap X'_i)$. We have:
    \begin{align*}
        v_i(X'_i) &\geq v_i(X_i) && (\text{by the design of the algorithm}) \\
        &\geq \alpha \mu_i(2, X_i \cup X_j') && (\text{since $X$ is $\alpha$-PMMS}) \\
        &\geq \alpha\left(  v_i(B) - v_i(B \cap X_i') \right) && (\text{by the assumption of Case 1}) \\
        &\geq \alpha\left( \mu' - \tfrac{1}{2}v_i(X'_i)\right) && (\text{since $v_i(B) \geq \mu'$ and $v_i(B \cap X_i') \leq \tfrac{1}{2}v_i(X'_i)$})
    \end{align*}
    Rearranging gives
    \begin{align*}
        \frac{1}{\alpha}v_i(X'_i) + \frac{1}{2} v_i(X'_i) \geq \mu'
        \qquad \Longrightarrow \qquad v_i(X'_i) \geq \frac{2\alpha}{2+\alpha}\mu'
    \end{align*}
    establishing Equation~\eqref{eq:goalcycles}.

    \textbf{Case 2}: $ \mu_i(2, X_i \cup X'_j) \geq v_i(X_i) + (v_i(A) - v_i(A \cap X'_i))$. We have:
    \begin{align*}
        v_i(X_i) &\geq \alpha \mu_i(2, X_i \cup X'_j) && (\text{since $X$ is $\alpha$-PMMS}) \\
        &\geq \alpha \left( v_i(X_i) + (v_i(A) - v_i(A \cap X'_i)) \right)  && (\text{by the assumption of Case 2})\\
        &\geq \alpha \left( v_i(X_i)  + \mu' - v_i(X'_i) \right)  && (\text{since $v_i(A) \geq \mu'$ and $v_i(A \cap X_i') \leq v_i(X'_i)$})
    \end{align*}
        Since $v_i(X'_i) \geq v_i(X_i)$ by the design of the algorithm, rearranging gives
    \begin{align*}
     (1-\alpha)v_i(X'_i) \geq \alpha(\mu' - v_i(X'_i)) 
        \qquad \Longrightarrow \qquad v_i(X'_i) \geq \alpha \mu'.
    \end{align*}
        Since $\alpha \geq \frac{2\alpha}{2+\alpha}$, this implies Equation~\eqref{eq:goalcycles}, completing the proof.
\end{proof}

\subsection{The Lone Divider Approach Does Not Guarantee EPMMS}

The lone divider approach \cite{steinhaus,kuhn1967games} is a common technique for achieving share-based fairness notions, including approximations of the maximin share \cite{HosseiniSS22,Aigner-HorevS22,Hummel25,HummelH26,AkramiR25a}. The algorithm proceeds in rounds. In each round, a designated ``divider'' agent partitions the remaining items into $n'$ bundles---where $n'$ is the number of remaining agents---such that each bundle meets her own fair share threshold. Next, a bipartite graph is constructed between the remaining agents and these bundles, with an edge connecting an agent to a bundle whenever the bundle meets that agent's fair share threshold. A (partial or perfect) matching is then computed in this graph, chosen so that no unmatched agent has an edge to a matched bundle. Finally, each matched agent is assigned the bundle to which she is paired, and these agents and bundles are removed before proceeding to the next round.

To establish the existence of EEFX allocations for general monotone valuations, \citet*{AkramiN24} elegantly adapt this approach by defining desirability through the EEFX property rather than a numerical share. Specifically, for any partition of the items, an \emph{EEFX-graph} is constructed in which an edge exists between an agent and a bundle if and only if the bundle is EEFX-feasible for that agent; see the following definition.

\begin{definition}[EEFX-Feasibility, \cite{AkramiN24}]
\label{def:n_epistemic_efx}
Let $k$ be a positive integer, $i \in [n]$ an agent, and $S \subseteq M$ a subset of items. A bundle $A \subseteq S$ is \emph{$k$-EEFX-feasible} for agent $i$ with respect to $S$ if there exists a partition of $S \setminus A$ into $k - 1$ bundles $C_1, C_2, \dots, C_{k-1}$ such that
\[
v_i(A) \geq v_i(C_j \setminus \{g\}) \quad \text{for every } j \in [k-1] \text{ and every } g \in C_j.
\]
We then define
\[
\textsc{EEFX}^k_i(S) = \{A \subseteq S \mid A \text{ is $k$-EEFX-feasible for agent $i$ with respect to $S$}\}.
\]
\end{definition}

The lone divider approach succeeds in this setting because of the following structural property: if a bundle is EEFX-feasible for an agent once some items have been removed, then it was already EEFX-feasible before their removal. Consequently, after each round it suffices to verify the EEFX condition on the reduced instance, without revisiting bundles assigned in earlier rounds. The following lemma of \citet*{AkramiN24} formalizes this property.

\begin{lemma}[{\cite[Lemma 3.1]{AkramiN24}}]\label{lem:EEFXreduce}
    Consider any fair division instance, an agent $i \in [n]$, and a bundle $A \subseteq M$ with $A \notin \textsc{EEFX}^n_i(M)$. Then every $B \in \textsc{EEFX}^{n-1}_i(M \setminus A)$ satisfies $B \in \textsc{EEFX}^n_i(M)$.
\end{lemma}

While this property enables the modified lone divider method for EEFX, we now show that it fails when EEFX is replaced with EPMMS, ruling out a direct adaptation of the same argument. We begin with the analogous notion of EPMMS-feasibility.

\begin{definition}[EPMMS-Feasibility]
\label{def:epmms_feasibility}
Let $k$ be a positive integer, $i \in [n]$ an agent, and $S \subseteq M$ a subset of items. A bundle $A \subseteq S$ is \emph{$k$-EPMMS-feasible} for agent $i$ with respect to $S$ if there exists a partition of $S \setminus A$ into $k - 1$ bundles $C_1, C_2, \dots, C_{k-1}$ such that
\[
v_i(A) \geq \mu_i(2, A \cup C_j) \quad \text{for every } j \in [k-1].
\]
We then define
\[
\textsc{EPMMS}^k_i(S) = \{A \subseteq S \mid A \text{ is $k$-EPMMS-feasible for agent $i$ with respect to $S$}\}.
\]
\end{definition}

The following lemma shows that the analogue of Lemma~\ref{lem:EEFXreduce} fails for EPMMS.

\begin{lemma}
There exists a fair division instance, an agent $i \in [n]$, and a bundle $A \subseteq M$ with $A \notin \textsc{EPMMS}^n_i(M)$ such that some $B \in \textsc{EPMMS}^{n-1}_i(M \setminus A)$ satisfies $B \notin \textsc{EPMMS}^n_i(M)$.
\end{lemma}

\begin{proof}
Consider an instance with $n = 3$ agents and six items, and let $i$ be an agent with an additive valuation assigning values $1, 4, 4, 6, 6, 6$ to the six items. Throughout the proof, we represent bundles by the multisets of their item values, and we drop the subscript $i$ since only one agent is involved.

Consider the partition
\[
X = (X_1, X_2, X_3) = (\{6\}, \{4,4,1\}, \{6,6\}).
\]
We first show that $X_2 \notin \textsc{EPMMS}^3(M)$. Suppose for contradiction that some partition $(C_1, C_2)$ of $M \setminus X_2 = \{6,6,6\}$ into two bundles satisfies $v(X_2) \geq \mu(2, X_2 \cup C_j)$ for both $j \in \{1, 2\}$. Without loss of generality, $C_1$ contains at least two items of value $6$, so
\[
\mu(2, X_2 \cup C_1) \geq \mu(2, \{1,4,4,6,6\}) = \min\{6+4+1,\; 6+4\} = 10 > 9 = v(X_2),
\]
a contradiction. Hence $X_2 \notin \textsc{EPMMS}^3(M)$.

Now consider $B = \{6\}$ in the reduced instance $M \setminus X_2 = \{6,6,6\}$. The partition $(\{6\}, \{6,6\})$ certifies that $B \in \textsc{EPMMS}^2(M \setminus X_2)$, since
\[
v(B) = 6 \geq \mu(2, \{6,6,6\}) = \min\{6+6,\; 6\} = 6.
\]

We now argue that $B \notin \textsc{EPMMS}^3(M)$. Suppose for contradiction that some partition $(D_1, D_2)$ of $M \setminus B = \{1,4,4,6,6\}$ into two bundles satisfies $v(B) \geq \mu(2, B \cup D_j)$ for both $j \in \{1, 2\}$. Since $|M \setminus B| = 5$, at least one of $D_1, D_2$ contains three or more items; without loss of generality, assume $|D_1| \geq 3$. Then moving the lowest-valued item of $D_1$ into $B$ yields two bundles each of value strictly greater than $v(B) = 6$, so $\mu(2, B \cup D_1) > v(B)$, a contradiction. Hence $B \notin \textsc{EPMMS}^3(M)$, completing the proof.
\end{proof}

\subsection{EPMMS Is Incomparable to Both MMS and EFX}\label{sec:mms_and_pmms}

In this section, we show that EPMMS neither implies nor is implied by MMS or EFX. This stands in contrast to EEFX, which is implied by both MMS and EFX~\cite{CaragiannisGRSV25}.
We first show that EPMMS is incomparable to MMS.

\begin{observation}
\label{MMS_and_EPMMS_incomparable}
Neither EPMMS nor MMS implies the other, even when agents have identical bivalued valuations (a special case of additive valuations).
\end{observation}

\begin{proof}
$(\text{EPMMS} \nRightarrow \text{MMS})$
Consider an instance with $6$ items and agents having identical bivalued valuations that assign value $1$ to three items and value $2$ to the remaining three. Consider the allocation
\[
X = (X_1, X_2, X_3) = (\{2\}, \{2,2\}, \{1,1,1\}),
\] 
where each bundle is represented by the values of the items it contains. Note that the $3$-maximin share is
\begin{align*}
    \mu(3,M) = 3 = \min\{ 2 + 1,\ 2 + 1,\ 2 + 1 \}.
\end{align*}
Hence $X$ is not an MMS allocation, as agent $1$ receives value $2 < 3$. 

We claim, however, that $X$ is PMMS, and therefore EPMMS. Since agent $2$ does not envy agents $1$ and $3$, and agent $3$ does not envy agent $1$, the PMMS condition is automatically satisfied for these pairs (Proposition~\ref{prop:ef_implies_pmms}). For the remaining pairs, we verify directly:
\begin{align*}
    v(X_1) &= 2 = \mu(2, X_1 \cup X_2) = \mu(2, \{2,2,2\}), \\
    v(X_1) &= 2 = \mu(2, X_1 \cup X_3) = \mu(2, \{2,1,1,1\}), \\
    v(X_3) &= 3 = \mu(2, X_2 \cup X_3) = \mu(2, \{2,2,1,1,1\}).
\end{align*}
This completes the proof of the first part of the observation.

$(\text{MMS} \nRightarrow \text{EPMMS})$. 
Consider an instance with $9$ items and agents having identical bivalued valuations that assign value $5$ to two items and value $3$ to the remaining seven. Consider the allocation
\[
X = (X_1, X_2, X_3) = (\{3,3,3\}, \{5,5\}, \{3,3,3,3\}),
\] 
where bundles are represented by the values of the items they contain. Observe that the $3$-maximin share equals
\begin{align*}
    \mu(3,M) = 9 &= \min\{ 3 + 3 + 3,\ 5 + 5,\ 3 + 3 + 3 + 3 \} \\
    &= \min\{ 3 + 3 + 3,\ 5 + 3 + 3,\ 5 + 3 + 3 \}.
\end{align*}
Therefore, $X$ is an MMS allocation, since each agent receives value at least $9$.

However, $X$ is not an EPMMS allocation. Suppose, for the sake of contradiction, that there exists an allocation $X' = (X'_1, X'_2, X'_3)$ such that $X'_2 = X_2$ and agent $2$ does not PMMS-envy any other agent under $X'$. Notice that either $X'_1$ or $X'_3$ must contain at least four items of value $3$. Without loss of generality, suppose this holds for $X'_1$. Then
\[
    \mu(2, X'_1 \cup X'_2) \geq \mu(2, \{3,3,3,3\} \cup \{5,5\}) = 11 = \min\{5 + 3 + 3,\ 5 + 3 + 3\}.
\]
Hence agent $2$ PMMS-envies agent $1$, a contradiction. We conclude that MMS does not imply EPMMS.
\end{proof}

We now show that EPMMS is incomparable to EFX.

\begin{observation}
Neither EPMMS nor EFX implies the other, even when agents have identical bivalued valuations (a special case of additive valuations).
\end{observation}
\begin{proof}
    $(\text{EPMMS} \nRightarrow \text{EFX})$
   Consider an instance with $2$ items and $3$ agents, where every agent values both items at $1$. Consider the allocation
    \[
        X = (X_1, X_2, X_3) = (\{1,1\}, \emptyset, \emptyset),
    \]
    where bundles are represented by the values of the items they contain. This allocation is not EFX, since agents $2$ and $3$ envy agent $1$ even after removing any single good. However, $X$ is EPMMS: both agents $2$ and $3$ could hypothetically repartition the two items assigned to agent $1$, splitting them equally between the remaining agents to obtain a PMMS allocation.

    $(\text{EFX} \nRightarrow \text{EPMMS})$
    Consider an instance with $5$ items and two agents having identical additive valuations that assign value $3$ to two items and value $2$ to the remaining three. Consider the allocation
    \[
        X = (X_1, X_2) = (\{3,2,2\}, \{3,2\}),
    \]
    where bundles are represented by the values of the items they contain. This allocation is EFX, since agent $2$ does not envy agent $1$ after the removal of any single item. However, $X$ is not EPMMS: in the two-agent case, EPMMS coincides with PMMS, and 
    \[
        \mu(2, \{3,2,2\} \cup \{3,2\}) = \min\{ 3 + 3,\ 2 + 2 + 2\} = 6,
    \]
    which is strictly larger than $3 + 2 = 5$.
\end{proof}

\section{Separation between EGMMS and EPMMS}\label{sec:sep_egmms_epmms}

In this section, we show that although EGMMS directly implies both EPMMS and MMS, the converse does not hold. In fact, we prove a stronger separation: even simultaneously satisfying MMS and PMMS does not imply EGMMS.

\begin{proposition}\label{prop:mms_and_pmms_does_not_imply_egmms}
There exists an instance with identical additive valuations admitting an allocation that is both MMS and PMMS, yet \emph{not} EGMMS.
\end{proposition}

\begin{proof}
  $(\text{PMMS} \land \text{MMS} \nRightarrow \text{EGMMS})$
Consider four agents with identical additive valuations over eight items of values
\[
5,\;4,\;4,\;3,\;2,\;2,\;2,\;1,
\]
and the allocation
\[
X_1=\{1,4\},\qquad X_2=\{5\},\qquad X_3=\{2,2,2\},\qquad X_4=\{3,4\},
\]
where each bundle is written as the multiset of item values.

We first argue that the allocation is MMS. The total value is $23$, so the MMS value is at most $\lfloor 23/4\rfloor = 5$. Every agent receives value at least $5$, so the allocation is MMS.

Next, we argue that the allocation is PMMS. Since the valuations are additive, for any pair $i,j$ with $v(X_i)\ge v(X_j)$ the PMMS condition for agent~$i$ holds (Proposition~\ref{prop:ef_implies_pmms}). It therefore suffices, for each pair, to verify the condition for the agent with the smaller bundle:
\begin{align*}
\mu(2,\, X_1\cup X_3) &= \mu(2,\{1,2,2,2,4\}) \le \lfloor 11/2\rfloor = 5 = v(X_1),\\
\mu(2,\, X_1\cup X_4) &= \mu(2,\{1,3,4,4\}) = \min\{4+1,\;4+3\} = 5 = v(X_1),\\
\mu(2,\, X_2\cup X_3) &= \mu(2,\{2,2,2,5\}) \le \lfloor 11/2\rfloor = 5 = v(X_2),\\
\mu(2,\, X_2\cup X_4) &= \mu(2,\{3,4,5\}) = \min\{3+4,\;5\} = 5 = v(X_2),\\
\mu(2,\, X_3\cup X_4) &= \mu(2,\{2,2,2,3,4\}) \le \lfloor 13/2\rfloor = 6 = v(X_3).
\end{align*}
Hence the allocation is PMMS.

We now argue that the allocation is not EGMMS. We show that agent~$1$ admits no EGMMS-certificate. Suppose, for contradiction, that
\[
Y = (Y_1, Y_2, Y_3, Y_4),\qquad Y_1 = X_1 = \{1,4\},
\]
is an EGMMS-certificate for agent~$1$; that is, $Y$ is a partition of all items with $Y_1 = X_1$, and for every subgroup $S$ with $1 \in S$, it holds that
\[
v(X_1)\;\ge\;\mu\bigl(|S|,\; \textstyle\bigcup_{j\in S} Y_j\bigr).
\]

We claim that the pairwise (i.e., $|S|=2$) conditions force the partition of the remaining items
$\{2,2,2,3,4,5\}$ uniquely (up to relabeling of $Y_2,Y_3,Y_4$). 

Assume without loss of generality that $5\in Y_2$. We claim
$Y_2=\{5\}$. If instead some item of value $a\ne 5$ also belongs to
$Y_2$, then $a\ge 2$, and the partition $(\{4,a\}, \{1,5\})$ of
$Y_1\cup Y_2$ shows
\[
\mu(2,\, Y_1\cup Y_2) \;\ge\; \min\{4+a,\,5+1\} \;\ge\; 6 \;>\; 5 \;=\; v(X_1),
\]
a contradiction. Hence $Y_2=\{5\}$.

Next, assume without loss of generality that $3\in Y_3$. We claim that
$Y_3 = \{3,4\}$.

Suppose first that $|Y_3|\ge 3$. Then $Y_3$ contains $3$ and two further
items $a,b$, each of value at least $2$. The partition
$(\{4,a\},\{1,3,b\})$ of $Y_1\cup Y_3$ gives
\[
\mu(2,\, Y_1\cup Y_3) \;\ge\; \min\{4+a,\,1+3+b\} \;\ge\; 6 \;>\; v(X_1),
\]
a contradiction. Suppose instead that $|Y_3|=1$. Then $Y_3=\{3\}$ and
$Y_4=\{2,2,2,4\}$, so the partition $(\{4,2,1\},\{2,2,4\})$ of
$Y_1\cup Y_4$ gives
\[
\mu(2,\, Y_1\cup Y_4) \;\ge\; \min\{4+2+1,\,2+2+4\} \;=\; 7 \;>\; v(X_1),
\]
a contradiction. Finally, suppose $Y_3=\{3,2\}$. Then $Y_4=\{2,2,4\}$,
and the partition $(\{4,2,1\},\{2,4\})$ of $Y_1\cup Y_4$ gives
\[
\mu(2,\, Y_1\cup Y_4) \;\ge\; \min\{4+2+1,\,2+4\} \;=\; 6 \;>\; v(X_1),
\]
a contradiction.
 
This confirms that the unique certificate (up to permutation) is
\[
Y_2=\{5\},\qquad Y_3=\{3,4\},\qquad Y_4=\{2,2,2\}.
\]

Now consider the subgroup $S=\{1,3,4\}$. The corresponding union is
\[
Y_1\cup Y_3\cup Y_4 \;=\; \{1,2,2,2,3,4,4\},
\]
which can be partitioned into three bundles of value $6$ each:
\[
\{2,4\},\qquad \{2,4\},\qquad \{1,2,3\}.
\]
Therefore $\mu(3,\, Y_1\cup Y_3\cup Y_4)\ge 6 > 5 = v(X_1)$, violating the groupwise requirement for $S$.

Since the unique pairwise-feasible certificate fails the groupwise condition, no EGMMS-certificate for agent~$1$ exists, and the allocation is not EGMMS.
\end{proof}

\section{Relations Between the Valuation Classes}\label{sec:valuation_classes}

\subsection{GMMS-Feasible Valuations Are A Strict Subclass of PMMS-feasible Valuations}

We prove that PMMS-feasible valuations form a strict subclass of GMMS-feasible valuations. That every GMMS-feasible valuation is also PMMS-feasible follows directly from the definitions. The converse, however, does not hold, as we demonstrate below.

\begin{proposition}\label{thm:mmsfandgmmsf}
    There exists a monotone PMMS-feasible valuation function that is not GMMS-feasible.
\end{proposition}
\begin{proof}
$(\text{PMMS-feasibility} \nRightarrow \text{GMMS-feasibility})$
Consider a set of $9$ items. Define a valuation function $v$ such that $v(S) = 0$ for all subsets $S$ with $|S| \le 2$ and $v(S) = 1$ for all subsets with $|S| \ge 4$.
For subsets of size $3$, the value is defined as follows:
\begin{align*}
   v(S) = 0 \text{ for } S \in \{
&\{1,2,5\}, \{1,2,9\}, \{1,3,4\}, \{1,3,5\}, \{1,3,6\}, \{1,3,7\}, \{1,3,8\},  \\
&\{1,3,9\}, \{1,4,5\}, \{1,4,7\}, \{1,4,8\}, \{1,4,9\}, \{1,5,6\}, \{1,5,7\},  \\
&\{1,5,8\}, \{1,5,9\}, \{1,6,9\}, \{1,7,8\}, \{1,7,9\}, \{1,8,9\}, \{2,3,5\},  \\
&\{2,5,8\}, \{2,5,9\}, \{3,4,5\}, \{3,4,9\}, \{3,5,6\}, \{3,5,7\}, \{3,5,8\},  \\
&\{3,5,9\}, \{3,6,9\}, \{3,7,9\}, \{3,8,9\}, \{4,5,7\}, \{4,5,8\}, \{4,5,9\}, \\
& \{4,8,9\},  \{5,6,8\}, \{5,6,9\}, \{5,7,8\}, \{5,7,9\}, \{5,8,9\} \}
\end{align*}
and $v(S) = 1$ for all other subsets $S$ of size $3$.
This valuation is not GMMS-feasible, since:
\begin{align*}
    \min \{ v(\{1,2,3\}), v(\{4,5,6\}), v(\{7,8,9\}) \} &= 1, \\ \max \{ v(\{1,4,7\}), v(\{2,5,8\}), v(\{3,6,9\}) \} &= 0
\end{align*}
It can, however, be verified that $v$ is PMMS-feasible. To check this, we verify that there is no pair of partitions $(X,Y)$ and $(A,B)$ satisfying $X \cap Y = \emptyset$, $A \cap B = \emptyset$, $X \cup Y = A \cup B$, and $\min\{v(X), v(Y)\} > \max\{v(A), v(B)\}$. The critical case to inspect is when all sets have size $3$, i.e. $|X| = |Y| = |A| = |B| = 3$, because for $\min\{v(X), v(Y)\} > \max\{v(A), v(B)\}$ to hold, we must have $v(X) = v(Y) = 1$ (implying $|X|, |Y| \ge 3$) and $v(A) = v(B) = 0$ (implying $|A|, |B| \le 3$). The only possible configuration satisfying these size constraints is when all sets have size $3$, and these cases can be verified computationally. There are $\binom{9}{6} = 84$ subsets of size $6$ in a $9$-element set. Each such subset contains $\binom{6}{3} = 20$ subsets of size $3$, which form $\binom{20}{2} = 190$ distinct pairs to verify. Hence, there are in total $84 \times 190 = 15{,}960$ cases to check, which we have verified computationally. 
\end{proof}

\subsection{Nice Cancelable Valuations Are A Strict Subclass of GMMS-Feasible Valuations}

We prove that GMMS-feasible valuations form a strict superset of nice cancelable valuations, a class known to include additive, budget-additive, unit-demand, and multiplicative valuations \cite{BergerCFF2021}. We begin by defining the class of nice cancelable valuations.

\begin{definition}[Nice Cancelable Valuations, \cite{BergerCFF2021}]
A valuation $v_i : 2^M \to \mathbb{R}_{\geq 0}$ is \emph{cancelable} if $v_i(S \cup \{g\}) > v_i(T \cup \{g\})$ implies $v_i(S) > v_i(T)$ for all $S, T \subseteq M$ and every $g \in M \setminus (S \cup T)$. A valuation $v_i$ is \emph{nice cancelable} if it is cancelable and there exists a cancelable valuation $v'_i : 2^M \to \mathbb{R}_{\geq 0}$ that is \emph{non-degenerate}, meaning $v'_i(S) \neq v'_i(T)$ for all distinct $S, T \subseteq M$, and \emph{respects} $v_i$, meaning $v_i(S) > v_i(T)$ implies $v'_i(S) > v'_i(T)$ for all $S, T \subseteq M$.
\end{definition}

Next, we show that nice-cancelability implies GMMS-feasibility.

\begin{proposition}
Every nice cancelable valuation is GMMS-feasible.
\end{proposition}

\begin{proof}
$(\text{Nice-cancelability} \Rightarrow \text{GMMS-feasibility})$
    Let $v$ be a nice cancelable valuation, and let $v'$ be a non-degenerate cancelable valuation that respects $v$. Fix any $S \subseteq M$, any $k \geq 2$, and any two partitions $A = (A_1, \ldots, A_k)$ and $B = (B_1, \ldots, B_k)$ of $S$. Suppose, for contradiction, that
    \begin{align*}
        \min\{v(A_1), \ldots, v(A_k)\} > \max\{v(B_1), \ldots, v(B_k)\}. 
    \end{align*}
    Since $v'$ respects $v$, this implies
    \begin{align}
        \min\{v'(A_1), \ldots, v'(A_k)\} > \max\{v'(B_1), \ldots, v'(B_k)\}. \label{eq:assumption}
    \end{align}

    We show by induction on $\ell$ that $v'(A_1 \cup \cdots \cup A_\ell) \geq v'(B_1 \cup \cdots \cup B_\ell)$ for all $1 \leq \ell \leq k - 1$. The base case $\ell = 1$ is immediate from Equation~\eqref{eq:assumption}. Assume the claim holds for some $1 \leq \ell < k - 1$. By Equation~\eqref{eq:assumption}, $v'(A_{\ell+1}) > v'(B_{\ell+1})$. Let $X = A_{\ell+1} \setminus (B_1 \cup \dots \cup B_{\ell})$ and $Y = (B_1 \cup \dots \cup B_{\ell}) \setminus A_{\ell+1}$.
    Applying the contrapositive of cancelability of $v'$ twice yields
    \begin{align*}
        v'(A_1 \cup \cdots \cup A_\ell) &\geq v'(B_1 \cup \cdots \cup B_\ell) \implies v'(A_1 \cup \cdots \cup A_{\ell} \cup X) \geq v'(B_1 \cup \cdots \cup B_\ell \cup X), \\
        v'(A_{\ell+1}) &\geq v'(B_{\ell+1}) \implies v'(Y  \cup A_{\ell+1}) \geq v'(Y  \cup B_{\ell+1}).
    \end{align*}
    Chaining these inequalities gives $v'(A_1 \cup \cdots \cup A_{\ell} \cup X) \geq v'(Y \cup B_{\ell+1})$. Letting $Z = A_{\ell+1} \cap (B_1 \cup \dots \cup B_{\ell})$, and applying the contrapositive of cancelability of $v'$ once more yields
    \begin{align*}
        v'(A_1 \cup \dots \cup A_{\ell+1}) = v'(A_1 \cup \cdots \cup A_{\ell} \cup X \cup Z) \geq v'(Y \cup B_{\ell+1} \cup Z) = v'(B_1 \cup \dots \cup B_{\ell+1})
    \end{align*}
    completing the induction.

    Taking $\ell = k - 1$, we obtain $v'(A_1 \cup \cdots \cup A_{k-1}) \geq v'(B_1 \cup \cdots \cup B_{k-1})$. Since $A$ and $B$ are partitions of the same set $S$, we have $A_1 \cup \cdots \cup A_{k-1} = S \setminus A_k$ and $B_1 \cup \cdots \cup B_{k-1} = S \setminus B_k$. Because $A_k \neq B_k$ (since, by Equation~\eqref{eq:assumption}, $v'(A_k) > v'(B_k)$), these two sets are distinct, and non-degeneracy of $v'$ implies strict inequality
    \begin{align}
        v'(S \setminus A_k) > v'(S \setminus B_k). \label{eq:toapply}
    \end{align}
    Writing $S \setminus A_k = (B_k \setminus A_k) \cup (S \setminus (A_k \cup B_k))$ and $S \setminus B_k = (A_k \setminus B_k) \cup (S \setminus (A_k \cup B_k))$, and applying cancelability of $v'$ to Equation~\eqref{eq:toapply} to remove the common elements of $S \setminus (A_k \cup B_k)$ one at a time, we obtain
    \[
        v'(B_k \setminus A_k) > v'(A_k \setminus B_k).
    \]
    Applying cancelability once more to add the common elements of $A_k \cap B_k$ one at a time gives
    \[
        v'(B_k) = v'((B_k \setminus A_k) \cup (A_k \cap B_k)) \geq v'((A_k \setminus B_k) \cup (A_k \cap B_k)) = v'(A_k).
    \]
    As before, $A_k \neq B_k$, so non-degeneracy of $v'$ implies strict inequality $v'(B_k) > v'(A_k)$. Since $v'$ respects $v$, this gives $v(B_k) \geq v(A_k)$, contradicting Equation~\eqref{eq:assumption}.
\end{proof}

Finally, we show that GMMS-feasibility does not necessarily imply nice-cancelability.

\begin{proposition}
    Not every GMMS-feasible valuation is nice cancelable.
\end{proposition}
\begin{proof}
$(\text{GMMS-feasibility} \nRightarrow \text{nice-cancelability})$
    \citet[Example 1]{AkramiACGMM2022} provide an example of a set of items $M = \{g_1, g_2, g_3\}$ and a monotone PMMS-feasible valuation function $v: 2^M \to \mathbb{R}_{\geq 0}$ that is not nice cancelable.
    We claim that the constructed valuation also satisfies GMMS-feasibility.
    We note that the condition of Definition~\ref{def:gmmsfeasible} holds for any partition with $k=2$ directly from PMMS-feasibility, and the case $k=1$ is trivial. Consider the case $k=3$, and let $A=(A_1, A_2, A_3)$ and $B = (B_1, B_2, B_3)$ be any two partitions of a subset $S \subseteq M$.
    Either (1) one of the subsets $B_1, B_2, B_3$ is empty, or (2) each of $B_1, B_2, B_3$ is a singleton. In case (1), we have $\min\{v(B_1), v(B_2), v(B_3)\} = 0$, so the condition holds trivially. In case (2), we have $\min\{v(B_1), v(B_2), v(B_3)\} = \min\{v(g_1), v(g_2), v(g_3)\}$, and the singleton of minimal value must belong to one of the sets $A_1, A_2, A_3$, which by monotonicity has value at least that of the singleton. Hence the condition holds.
\end{proof}

\subsection{When Does Envy-Freeness Imply GMMS or PMMS?}\label{sec:ef_implies_pmms}

We first show that PMMS-feasibility is both sufficient and necessary for envy-freeness to imply PMMS, so it precisely characterizes the valuations under which this implication holds.

\begin{proposition}\label{prop:ef_implies_pmms}
When $v_i$ is an PMMS-feasible valuation, any allocation that is envy-free for agent $i$ is also PMMS for agent $i$. Conversely, when $v_i$ is not PMMS-feasible, there exists a (possibly partial) allocation that is envy-free for agent $i$ but not PMMS for agent $i$.
\end{proposition}
\begin{proof}
$(\text{EF} \Rightarrow \text{PMMS under PMMS-feasibility})$ Let $X = (X_1, \ldots, X_n)$ be an allocation that is envy-free for agent $i$; that is, $v_i(X_i) \geq v_i(X_j)$ for every other agent $j$. Fix an agent $j \neq i$, and let $(X'_i, X'_j)$ be a partition of $X_i \cup X_j$ that attains the pairwise maximin share, so that
\[
    \mu_i(2, X_i \cup X_j) = \min\{v_i(X'_i),\, v_i(X'_j)\}.
\]
Then
\begin{align*}
    v_i(X_i) &= \max\{v_i(X_i),\, v_i(X_j)\} && \text{(by envy-freeness)} \\
    &\geq \min\{v_i(X'_i),\, v_i(X'_j)\} && \text{(by PMMS-feasibility)} \\
    &= \mu_i(2, X_i \cup X_j) && \text{(by the choice of $(X'_i, X'_j)$)}.
\end{align*}
Since $j$ was arbitrary, $X$ is PMMS for agent $i$.

$(\text{EF} \nRightarrow \text{PMMS without PMMS-feasibility})$  Assume $v_i$ is not PMMS-feasible. Then there exist two disjoint bundles $X_i, X_j \subseteq M$ and a partition $(X'_i, X'_j)$ of $X_i \cup X_j$ such that
\[
    \max\{v_i(X_i), v_i(X_j)\} < \min\{v_i(X'_i), v_i(X'_j)\}.
\]
Without loss of generality, assume $v_i(X_i) \geq v_i(X_j)$, and consider the (possibly partial) two-agent allocation $X = (X_i, X_j)$ among agents $i$ and $j$. This allocation is envy-free for agent $i$, since $v_i(X_i)  \geq  v_i(X_j)$. It is not, however, PMMS for agent $i$: 
\[
    v_i(X_i) = \max\{v_i(X_i), v_i(X_j)\} < \min\{v_i(X'_i), v_i(X'_j)\} \leq \mu_i(2, X_i \cup X_j),
\]
where the final inequality holds because $(X'_i, X'_j)$ is a particular partition of $X_i \cup X_j$, while $\mu_i(2, X_i \cup X_j)$ is the maximum of $\min\{v_i(A), v_i(B)\}$ over all such partitions. 
\end{proof}

Similarly, we show that GMMS-feasibility is both sufficient and necessary for envy-freeness to imply GMMS, so it precisely characterizes the valuations under which this implication holds.

\begin{proposition}
When $v_i$ is a GMMS-feasible valuation, any allocation that is envy-free for agent $i$ is also GMMS for agent $i$. Conversely, when $v_i$ is not GMMS-feasible, there exists a (possibly partial) allocation that is envy-free for agent $i$ but not GMMS for agent $i$.
    \end{proposition}
\begin{proof}
$(\text{EF} \Rightarrow \text{GMMS under GMMS-feasibility})$ Let $X = (X_1, \ldots, X_n)$ be an allocation that is envy-free for agent $i$; that is, $v_i(X_i) \geq v_i(X_j)$ for every other agent $j$. Fix a subset of agent $i \in J \subseteq N$, and let $(X'_j)_{j \in J}$ be a partition of $\bigcup_{j \in J} X_j$ that attains the groupwise maximin share, so that
\[
    \mu_i(|J|, \bigcup_{j \in J} X_j) = \min_{j \in J} v_i(X'_j).
\]
Then
\begin{align*}
    v_i(X_i) &= \max_{j \in J} v_i(X_j) && \text{(by envy-freeness)} \\
    &\geq \min_{j \in J} v_i(X'_j) && \text{(by GMMS-feasibility)} \\
    &= \mu_i(|J|, \bigcup_{j \in J} X_j) && \text{(by the choice of $(X'_j)_{j \in J}$)}.
\end{align*}
Since $J$ was arbitrary, $X$ is GMMS for agent $i$.

$(\text{EF} \nRightarrow \text{GMMS without GMMS-feasibility})$  Assume $v_i$ is not GMMS-feasible. Then there exist a set of indices $J$ with $i \in J$, disjoint bundles $(X_j)_{j \in J}$, and a partition $(X'_j)_{j \in J}$ of $\bigcup_{j \in J} X_j$ such that
\[
    \max_{j \in J} v_i(X_j) < \min_{j \in J} v_i(X'_j).
\]
Without loss of generality, assume $v_i(X_i) \geq v_i(X_j)$ for all $j \in J$, and consider the (possibly partial) allocation $X = (X_j)_{j \in J}$ among agents indexed by $J$, including agent $i$. This allocation is envy-free for agent $i$, since $v_i(X_i) \geq v_i(X_j)$ for all $j \in J$ by assumption. It is not, however, GMMS for $i$:
\[
    v_i(X_i) = \max_{j \in J} v_i(X_j) < \min_{j \in J} v_i(X'_j) \leq \mu_i(|J|,\, \bigcup_{j \in J} X_j),
\]
where the final inequality holds because $(X'_j)_{j \in J}$ is a particular partition of $\bigcup_{j \in J} X_j$, while $\mu_i(|J|, \bigcup_{j \in J} X_j)$ is the maximum of $\min_{j \in J} v_i(X''_j)$ over all such partitions $X''$. 
\end{proof}

\subsection{GMMS-Feasible Valuations Are Incomparable with Gross Substitutes Valuations}

We show that (G)PMMS-feasible valuations are incomparable with the class of gross substitutes valuations \cite{Leme17}, an important and well-studied member of the hierarchy of complement-free valuations \cite{LehmannLN06}. We begin by defining the class of OXS valuations \cite{LehmannLN06,BenabbouCEZ19,BenabbouCIZ20}, which forms a special case of gross substitutes valuations \cite{Leme17}.

\begin{definition}[OXS Valuations, \cite{LehmannLN06}]
A valuation function $v: 2^M \to \mathbb{R}_{\geq 0}$ over a set of items $M$ is an {OXS valuation} if there exists a finite set $X$ and a weight function $w: M \times X \to \mathbb{R}_{\geq 0}$ such that for every subset $S \subseteq M$:
\begin{align*}
    v(S) = \max_{\pi \in \mathcal{M}(S, X)} \sum_{(j, k) \in \pi} w({j,k})
\end{align*}
where $\mathcal{M}(S, X)$ is the set of all matchings between the items in $S$ and the elements of $X$.
\end{definition}

We first show that not every gross substitutes valuation is PMMS-feasible.

\begin{proposition}\label{prop:gsmmsfeasilbe}
There exists an OXS valuation function, and consequently a gross substitutes valuation function, that is not PMMS-feasible.
\end{proposition}
\begin{proof}
$(\text{OXS} \nRightarrow \text{PMMS-feasibility})$
Let $M = \{a,b,c,d\}$ denote the set of items, and let $X = \{1,2\}$. Define the weight function $w : M \times X \to \mathbb{R}_{\geq 0}$ by setting $w(a,1) = w(b,1) = w(c,2) = w(d,2) = 1$ and $w(i,j) = 0$ for all remaining pairs; see Figure~\ref{fig:matching} for the corresponding bipartite graph.
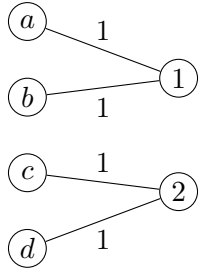
\begin{figure}[t]
\centering
\begin{tikzpicture}[
vertex/.style={circle,draw,minimum size=5mm,inner sep=1pt},
every edge node/.style={font=\small,fill=none,draw=none}
]
\node[vertex] (a) at (0,1.5) {$a$};
\node[vertex] (b) at (0,0.5) {$b$};
\node[vertex] (c) at (0,-0.5) {$c$};
\node[vertex] (d) at (0,-1.5) {$d$};
\node[vertex] (one) at (2,0.75) {$1$};
\node[vertex] (two) at (2,-0.75) {$2$};
\draw (a) -- node[midway,above] {1} (one);
\draw (b) -- node[midway,below] {1} (one);
\draw (c) -- node[midway,above] {1} (two);
\draw (d) -- node[midway,below] {1} (two);
\end{tikzpicture}
\caption{The weight function $w$ used in the proof of Proposition~\ref{prop:gsmmsfeasilbe}. Edges represent pairs with weight $1$; all omitted pairs have weight $0$.}\label{fig:matching}
\end{figure}

Let $v$ be the OXS valuation function induced by $w$. This valuation is not PMMS-feasible, since
\begin{align*}
\min \{ v(\{a,c\}), v(\{b,d\}) \} = 2 > 1 = \max \{ v(\{a,b\}), v(\{c,d\}) \}.
\end{align*}
This completes the proof.
\end{proof}

Finally, the analysis of \citet[Section 2]{Leme17} shows that not every GMMS-feasible valuation is a gross substitutes valuation.
\begin{proposition}[\cite{Leme17}]\label{thm:mmsfeasilbegs}
    There exists a budget-additive valuation function, and consequently a GMMS-feasible valuation, that is not a gross substitutes valuation.
\end{proposition}

\section{Tightness of Our Analysis of the Envy Cycles Procedure}\label{sec:tightness_envy_cycles}

In this section, we show that our analysis of the Envy Cycle Elimination algorithm (Proposition~\ref{prop:4_5_PMMS_identical_rankings_1}) is tight: the algorithm guarantees no better than $4/5$-PMMS.

\begin{proposition}
    There exists an instance where the output of Algorithm~\ref{alg:envy_cycles} does not satisfy $\alpha$-PMMS for any $\alpha > \tfrac{4}{5}$.
\end{proposition}
\begin{proof}
Assume there are three agents and nine items. Let $\epsilon > 0$ be arbitrarily small, and define $\epsilon_k = \epsilon \cdot \frac{1}{2^{8-k}}$. The value functions are as follows:
\begin{center}
\begin{tabular}{ c c c c c c c c c c}
  & $g_1$ & $g_2$ & $g_3$ & $g_4$ & $g_5$ & $g_6$ & $g_7$ & $g_8$ & $g_9$\\ 
 \text{Agent 1} & $3$ & $2 + \epsilon_8$ & $2 + \epsilon_7$ & $2 + \epsilon_6$ & $\epsilon_5$ & $\epsilon_4$ & $\epsilon_3$ & $\epsilon_2$ & $\epsilon_0$\\  
 \text{Agent 2} & $2 + \epsilon_8$ & $2 + \epsilon_7$ & $2 + \epsilon_6$ & $2 + \epsilon_5$ & $2 + \epsilon_4$ & $2 + \epsilon_3$ & $1 + \epsilon_2$ & $1$ & $\epsilon_0$\\   
 \text{Agent 3} & $2 + \epsilon_8$ & $2 + \epsilon_6 +  \epsilon_1$ & $2 + \epsilon_6$ & $2 + \epsilon_5$ & $2 + \epsilon_4$ & $2 + \epsilon_3$ & $1 + \epsilon_2$ & $1$ & $\epsilon_0$ \\  
\end{tabular}
\end{center}
The perturbations $\epsilon_k$ are introduced to ensure that the instance is non-degenerate (as defined above).

We now simulate the algorithm on this instance. We also include the envy graph (whose nodes are the agents and which contains an edge $i \to j$ whenever $v_i(X_j) > v_i(X_i)$) at the end of each step.

After the first three iterations of the loop, the allocation will be as follows: 
\[
\begin{aligned}
X_1 &= \{g_1\}, & v_1(X_1) &= 3, \\
X_2 &= \{g_2\}, & v_2(X_2) &= 2 + \epsilon_7, \\
X_3 &= \{g_3\}, & v_3(X_3) &= 2 + \epsilon_6 \ \ \ \ \ \ \ \ \ \ \ \ \ \ \ \ \ \ \ \ \ \ \ \ 
\end{aligned}
\quad \quad
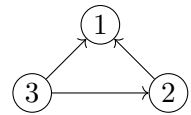
\begin{tikzpicture}[
    baseline={(current bounding box.center)},
    scale=1.2,
    main/.style={draw,circle,inner sep=2pt}
]
\node[main] (1) {$1$};
\node[main] (2) [below right=0.75cm of 1] {$2$};
\node[main] (3) [below left=0.75cm of 1] {$3$};
\draw[->] (2)--(1);
\draw[->] (3)--(1);
\draw[->] (3)--(2);
\end{tikzpicture}
\]

Next, agent 3 receives an item:
\[
\begin{aligned}
X_1 &= \{g_1\}, & v_1(X_1) &= 3, \\
X_2 &= \{g_2\}, & v_2(X_2) &= 2 + \epsilon_7, \\
X_3 &= \{g_3, g_4\}, & v_3(X_3) &= 4 + \epsilon_6  + \epsilon_5 \ \ \ \ \ \ \ \ \ \ \ \   
\end{aligned}
\quad \quad
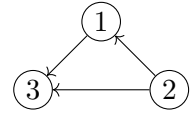
\begin{tikzpicture}[
    baseline={(current bounding box.center)},
    scale=1.2,
    main/.style={draw,circle,inner sep=2pt}
]
\node[main] (1) {$1$};
\node[main] (2) [below right=0.75cm of 1] {$2$};
\node[main] (3) [below left=0.75cm of 1] {$3$};
\draw[->] (1)--(3);
\draw[->] (2)--(3);
\draw[->] (2)--(1);
\end{tikzpicture}
\]

Next, agent 2 receives an item:

\[
\begin{aligned}
X_1 &= \{g_1\}, & v_1(X_1) &= 3, \\
X_2 &= \{g_2, g_5\}, & v_2(X_2) &= 4 + \epsilon_7 + \epsilon_4, \\
X_3 &= \{g_3, g_4\}, & v_3(X_3) &= 4 + \epsilon_6  + \epsilon_5 \ \ \ \ \ \ \ \ \ \ \ \ 
\end{aligned}
\quad \quad
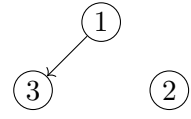
\begin{tikzpicture}[
    baseline={(current bounding box.center)},
    scale=1.2,
    main/.style={draw,circle,inner sep=2pt}
]
\node[main] (1) {$1$};
\node[main] (2) [below right=0.75cm of 1] {$2$};
\node[main] (3) [below left=0.75cm of 1] {$3$};
\draw[->] (1)--(3);
\end{tikzpicture}
\]

Next, agent 2 receives another item:

\[
\begin{aligned}
X_1 &= \{g_1\}, & v_1(X_1) &= 3, \\
X_2 &= \{g_2, g_5, g_6\}, & v_2(X_2) &= 6 + \epsilon_7 + \epsilon_4 + \epsilon_3, \\
X_3 &= \{g_3, g_4\}, & v_3(X_3) &= 4 + \epsilon_6  + \epsilon_5 
\end{aligned}
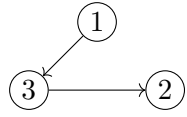
\quad \quad
\begin{tikzpicture}[
    baseline={(current bounding box.center)},
    scale=1.2,
    main/.style={draw,circle,inner sep=2pt}
]
\node[main] (1) {$1$};
\node[main] (2) [below right=0.75cm of 1] {$2$};
\node[main] (3) [below left=0.75cm of 1] {$3$};
\draw[->] (1)--(3);
\draw[->] (3)--(2);
\end{tikzpicture}
\]

Next, agent 1 receives an item:

\[
\begin{aligned}
X_1 &= \{g_1, g_7\}, & v_1(X_1) &= 3 + \epsilon_3, \\
X_2 &= \{g_2, g_5, g_6\}, & v_2(X_2) &= 6 + \epsilon_7 + \epsilon_4 + \epsilon_3, \\
X_3 &= \{g_3, g_4\}, & v_3(X_3) &= 4 + \epsilon_6  + \epsilon_5 
\end{aligned}
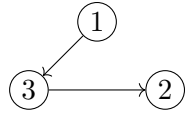
\quad \quad
\begin{tikzpicture}[
    baseline={(current bounding box.center)},
    scale=1.2,
    main/.style={draw,circle,inner sep=2pt}
]
\node[main] (1) {$1$};
\node[main] (2) [below right=0.75cm of 1] {$2$};
\node[main] (3) [below left=0.75cm of 1] {$3$};
\draw[->] (1)--(3);
\draw[->] (3)--(2);
\end{tikzpicture}
\]

Next, agent 1 receives another item:

\[
\begin{aligned}
X_1 &= \{g_1, g_7, g_8\}, & v_1(X_1) &= 3 + \epsilon_3 + \epsilon_2, \\
X_2 &= \{g_2, g_5, g_6\}, & v_2(X_2) &= 6 + \epsilon_7 + \epsilon_4 + \epsilon_3, \\
X_3 &= \{g_3, g_4\}, & v_3(X_3) &= 4 + \epsilon_6  + \epsilon_5 
\end{aligned}
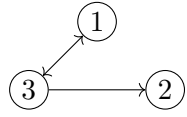
\quad \quad
\begin{tikzpicture}[
    baseline={(current bounding box.center)},
    scale=1.2,
    main/.style={draw,circle,inner sep=2pt}
]
\node[main] (1) {$1$};
\node[main] (2) [below right=0.75cm of 1] {$2$};
\node[main] (3) [below left=0.75cm of 1] {$3$};
\draw[<->] (1)--(3);
\draw[->] (3)--(2);
\end{tikzpicture}
\]

At this point, agent 3 begins to envy agent 1, since
\[
v_3(X_1) = 4 + \epsilon_8 + \epsilon_2 \ge 4 + \epsilon_6 + \epsilon_5 = v_3(X_3).
\]
We therefore resolve the cycle between agents 1 and 3:

\[
\begin{aligned}
X_1 &= \{g_3, g_4\}, & v_1(X_1) &= 4 + \epsilon_7 + \epsilon_6, \\
X_2 &= \{g_2, g_5, g_6\}, & v_2(X_2) &= 6 + \epsilon_7 + \epsilon_4 + \epsilon_3, \\
X_3 &= \{g_1, g_7, g_8\}, & v_3(X_3) &= 4 + \epsilon_8 + \epsilon_2
\end{aligned}
\quad \quad
\begin{tikzpicture}[
    baseline={(current bounding box.center)},
    scale=1.2,
    main/.style={draw,circle,inner sep=2pt}
]
\node[main] (1) {$1$};
\node[main] (2) [below right=0.75cm of 1] {$2$};
\node[main] (3) [below left=0.75cm of 1] {$3$};
\draw[->] (3)--(2);
\end{tikzpicture}
\]

Finally, agent 3 receives the last item:
\[
\begin{aligned}
X_1 &= \{g_3, g_4\}, & v_1(X_1) &= 4 + \epsilon_7 + \epsilon_6, \\
X_2 &= \{g_2, g_5, g_6\}, & v_2(X_2) &= 6 + \epsilon_7 + \epsilon_4 + \epsilon_3, \\
X_3 &= \{g_1, g_7, g_8, g_9\}, & v_3(X_3) &= 4 + \epsilon_8 + \epsilon_2 + \epsilon_0
\end{aligned}
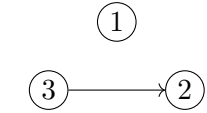
\quad \quad
\begin{tikzpicture}[
    baseline={(current bounding box.center)},
    scale=1.2,
    main/.style={draw,circle,inner sep=2pt}
]
\node[main] (1) {$1$};
\node[main] (2) [below right=0.75cm of 1] {$2$};
\node[main] (3) [below left=0.75cm of 1] {$3$};
\draw[->] (3)--(2);
\end{tikzpicture}
\]

The algorithm now terminates. Observe that
\[
v_3(X_3) = 4 + \epsilon_8 + \epsilon_2 + \epsilon_0,
\]
while
\[
\mu_3(2, X_2 \cup X_3) = \min\big(v_3(\{g_1,g_6,g_8\}),\, v_3(\{g_2,g_5,g_7,g_9\})\big)
= 5 + \epsilon_6 + \epsilon_4 + \epsilon_2 + \epsilon_1 + \epsilon_0.
\]
Thus, we cannot guarantee better than a $\tfrac{4}{5}$-PMMS allocation using Algorithm~\ref{alg:envy_cycles}.
\end{proof}










\section{EPMMS and EGMMS Reductions to Identical Rankings}
\label{sec:EPMMS_EGMMS_Reductions_framework}
\subsection{The Ordering Framework and Picking Sequence}
\label{sec:ordering_framework}

The ordering framework is a powerful tool in fair division, typically used to design approximation algorithms for Maximin Share (MMS) allocations. Recently, \citet*{CaragiannisGRSV25} utilized this framework to prove the existence of Epistemic EFX (EEFX) allocations for additive valuations. In this paper, we adapt their approach to establish guarantees for EPMMS approximations and EGMMS allocations under bivalued valuations. The core idea is to transform an instance with arbitrary valuations into a simplified instance where all agents share an identical ranking over the items, solve the allocation problem there, and then map the solution back to the original instance.

Let $\mathcal{I}$ be a fair division instance. For each agent $i \in [n]$, let $(\ell_1, \ell_2, \ldots, \ell_m)$ be an ordering of the items in $M$ sorted by $i$'s valuation in non-increasing order, such that:
\[
v_i(\ell_1) \ge v_i(\ell_2) \ge \cdots \ge v_i(\ell_m).
\]
Intuitively, $\ell_t$ represents agent $i$’s $t$-th favorite item. 

\begin{definition}[Item Rank]
For an agent $i \in [n]$ and an item $j \in M$, the \emph{rank} of $j$, denoted by $r_i(j)$, is the index $t$ such that $j = \ell_t$. For a set of items $S \subseteq M$, we define the set of their ranks as:
\[
r_i(S) := \{\, r_i(j) \mid j \in S \,\}.
\]
\end{definition}

\begin{definition}[Ordered Valuation and Instance]
The \emph{ordered valuation function} of agent $i$, denoted by $\hat{v}_i$, is defined for any subset of ranks $T \subseteq [m]$ as:
\[
\hat{v}_i(T) := v_i(\{\, \ell_t \mid t \in T \,\}).
\]
We denote by $\hat{\mathcal{I}} = \text{Order}(\mathcal{I})$ the ordered fair division instance, which contains the same agents and the same number of items as $\mathcal{I}$, but where each agent $i$ has the valuation function $\hat{v}_i$ over the items $[m]$.
\end{definition}

\begin{claim}
\label{claim:valuation_rank_equivalence}
For every subset of items $S \subseteq M$ and every agent $i \in [n]$, it holds that:
\[
v_i(S) = \hat{v}_i\bigl(r_i(S)\bigr).
\]
\end{claim}
\begin{example}
Suppose there are four items, $M = \{1, 2, 3, 4\}$, and agent $i$'s valuation is:
\[
v_i(1) = 6, \quad v_i(2) = 11, \quad v_i(3) = 8, \quad v_i(4) = 9.
\]
Their ranks are:
\[
r_i(1) = 4, \quad r_i(2) = 1, \quad r_i(3) = 3, \quad r_i(4) = 2.
\]
The ordered valuation function $\hat{v}_i$ evaluates the $t$-th best items directly:
\[
\hat{v}_i(1) = 11, \quad \hat{v}_i(2) = 9, \quad \hat{v}_i(3) = 8, \quad \hat{v}_i(4) = 6
\]
\end{example}

\subsubsection{The Reduction Algorithm}
To map allocations from the ordered instance $\hat{\mathcal{I}}$ back to the original instance $\mathcal{I}$, we employ a picking sequence algorithm. 

\begin{enumerate}
    \item \textbf{Ordering:} Construct the ordered instance $\hat{\mathcal{I}} = \text{Order}(\mathcal{I})$.
    \item \textbf{Base Allocation:} Compute the target allocation $\hat{X} = (\hat{X}_1, \ldots, \hat{X}_n)$ for the identical rankings instance $\hat{\mathcal{I}}$.
    \item \textbf{Picking Sequence:} Define the sequence $L = (a_1, a_2, \ldots, a_m)$ such that $a_k = i$ if the $k$-th item in the identical ranking is assigned to agent $i$ in $\hat{X}$.
    \item \textbf{Actual Allocation:} Execute the \textsc{Pick}$(\mathcal{I}, L)$ routine: in each step $k$, agent $a_k$ picks their favorite remaining item according to their original valuation $v_i$. Let $X$ be the resulting allocation.
\end{enumerate}

To analyze how values shift when translating $\hat{X}$ to $X$, we rely on a structural lemma introduced by Caragiannis et al \cite{CaragiannisGRSV25}.

\begin{lemma}[Lemma 2 in \cite{CaragiannisGRSV25}]
\label{lem:properties-of-bijection}
For every agent $i \in [n]$, there exists a bijection $\pi_i : [m] \to M$ such that:
\begin{itemize}
    \item For every $j \in \hat{X}_i$, we have $\pi_i(j) \in X_i$ and $r_i(\pi_i(j)) \le j$.
    \item For every $j \in [m] \setminus \hat{X}_i$, we have $\pi_i(j) \in M \setminus X_i$ and $r_i(\pi_i(j)) \ge j$.
\end{itemize}
\end{lemma}

Intuitively, Lemma~\ref{lem:properties-of-bijection} guarantees that when an agent transitions from their assigned bundle $\hat{X}_i$ to their actual bundle $X_i$, every item they receive \emph{weakly improves} in rank, while every item they do not receive \emph{weakly degrades} in rank.

Having established our ordering framework and the reduction algorithm, we are now ready to prove Proposition~\ref{prop:EPMMS_EGMMS_identical_rankings_WLOG}:
\begin{proof}[Proof for Proposition~\ref{prop:EPMMS_EGMMS_identical_rankings_WLOG}]
Let $\mathcal{I}$ be an arbitrary fair division instance, and let $\hat{\mathcal{I}}$ be the corresponding identical rankings instance obtained via the ordering framework. Let $\hat{X}$ be the $\alpha$-EPMMS (resp., $\alpha$-EGMMS) allocation in $\hat{\mathcal{I}}$, and let $X$ be the resulting allocation generated by the picking sequence algorithm.

Fix an arbitrary agent $i \in [n]$. Let $\pi_i : [m] \to [m]$ be the bijection guaranteed by Lemma~\ref{lem:properties-of-bijection}, mapping items from the ordered instance to the original instance.

Let $\hat{Y}^i = (\hat{Y}^i_1, \ldots, \hat{Y}^i_n)$ be the $\alpha$-EPMMS (resp., $\alpha$-EGMMS) certificate for agent $i$ with respect to the bundle $\hat{X}_i$ in $\hat{\mathcal{I}}$. By definition, $\hat{Y}^i_i = \hat{X}_i$. 
We construct the corresponding certificate $Y^i = (Y^i_1, \ldots, Y^i_n)$ in the original instance $\mathcal{I}$ by mapping the items through $\pi_i$:
\[ 
Y^i_k = \{\, \pi_i(g) \mid g \in \hat{Y}^i_k \,\} \quad \text{for all } k \in [n]. 
\]
Because $\pi_i$ is a bijection, $Y^i$ is a well-defined allocation of all items, and by the properties of $\pi_i$, $Y^i_i = X_i$. We must now show that $Y^i$ is a valid $\alpha$-EPMMS (resp., $\alpha$-EGMMS) certificate for agent $i$.

Let $J \subseteq [n]$ be a subset of agents such that $i \in J$. 
\begin{itemize}
    \item For $\alpha$-EPMMS, we restrict $J$ to subsets of size 2 (i.e., $J = \{i, j\}$ for some $j \neq i$). 
    \item For $\alpha$-EGMMS, $J$ can be any subset containing $i$. 
\end{itemize}

Let $k = |J|$ and define the combined bundle $S = \bigcup_{j \in J} Y^i_j$. Correspondingly, let $\hat{S} = \bigcup_{j \in J} \hat{Y}^i_j$ be the combined bundle in the ordered instance. 

Take an arbitrary $k$-partition $\mathcal{P} = (P_1, P_2, \ldots, P_k)$ of $S$. This partition naturally defines a corresponding $k$-partition $\hat{\mathcal{P}} = (\hat{P}_1, \hat{P}_2, \ldots, \hat{P}_k)$ of $\hat{S}$ where $\hat{P}_m = \{\, \pi_i^{-1}(g) \mid g \in P_m \,\}$. 

Because $\hat{Y}^i$ is an $\alpha$-EPMMS (resp., $\alpha$-EGMMS) certificate in the ordered instance, agent $i$'s valuation of $\hat{X}_i$ must be at least $\alpha$ times the value of the minimum piece in the partition $\hat{\mathcal{P}}$. Without loss of generality, assume $\hat{P}_1$ is this minimal piece. Therefore, we have:
\begin{align}
    \hat{v}_i(\hat{X}_i) \ge \alpha \hat{v}_i(\hat{P}_1). 
\end{align}

Let $R = X_i \cap P_1$, which implies $\hat{R} = \hat{X}_i \cap \hat{P}_1$. 
By the properties of the bijection $\pi_i$, items mapped into $X_i$ weakly improve in rank, while items mapped outside $X_i$ weakly degrade. This yields the following inequalities:
\begin{itemize}
    \item Inside $X_i$: $v_i(R) \ge \hat{v}_i(\hat{R})$ and $v_i(X_i \setminus R) \ge \hat{v}_i(\hat{X}_i \setminus \hat{R})$.
    \item Outside $X_i$: Because $P_1 \setminus R \ \cap \ X_i = \emptyset$, it follows that $v_i(P_1 \setminus R) \le \hat{v}_i(\hat{P}_1 \setminus \hat{R})$.
\end{itemize}

We now apply these bounds to the fairness condition. By additivity of valuations:
\begin{align*}
&\hat{v}_i(\hat{X}_i \setminus \hat{R}) + \hat{v}_i(\hat{R}) \ge \alpha \hat{v}_i(\hat{P}_1 \setminus \hat{R}) + \alpha \hat{v}_i(\hat{R}) \\
\implies & \hat{v}_i(\hat{X}_i \setminus \hat{R}) + (1 - \alpha)\hat{v}_i(\hat{R}) \ge \alpha \hat{v}_i(\hat{P}_1 \setminus \hat{R}).
\end{align*}
Since $0 < \alpha \le 1$, the coefficient $(1-\alpha)$ is non-negative. We can therefore substitute $\hat{v}_i(\hat{R})$ with the larger term $v_i(R)$ while preserving the inequality:
\begin{align*}
    \hat{v}_i(\hat{X}_i \setminus \hat{R}) + (1 - \alpha)v_i(R) \ge \alpha \hat{v}_i(\hat{P}_1 \setminus \hat{R}).
\end{align*}

Adding $\alpha v_i(R)$ to both sides, we obtain:
\begin{align*}
\hat{v}_i(\hat{X}_i \setminus \hat{R}) + v_i(R) \ge \alpha \hat{v}_i(\hat{P}_1 \setminus \hat{R}) + \alpha v_i(R). 
\end{align*}
Finally, we substitute the remaining ordered valuations with their original instance bounds:
\begin{align*}
v_i(X_i \setminus R) + v_i(R) & \ge \alpha \hat{v}_i(\hat{P}_1 \setminus \hat{R}) + \alpha v_i(R) && (\text{since } v_i(X_i \setminus R) \ge \hat{v}_i(\hat{X}_i \setminus \hat{R})) \\
&\ge \alpha v_i(P_1 \setminus R) + \alpha v_i(R) && (\text{since } v_i(P_1 \setminus R) \le \hat{v}_i(\hat{P}_1 \setminus \hat{R})) \\
 \implies  v_i(X_i) \ge \alpha v_i(P_1).  & && (\text{by additivity})
\end{align*}

Because $P_1$ was chosen such that $\hat{v}_i(\hat{X}_i) \ge \alpha \hat{v}_i(\hat{P}_1)$, and our derivation proves $v_i(X_i) \ge \alpha v_i(P_1)$, it trivially holds that:
\[ 
v_i(X_i) \ge \alpha \min_{P \in \mathcal{P}} v_i(P). 
\]
Since the $k$-partition $\mathcal{P}$ of the combined bundle $S$ was chosen arbitrarily, this lower bound holds for all possible $k$-partitions. Therefore, agent $i$ receives at least their $\alpha$-Maximin Share of the bundle $S$. Because this holds for an arbitrary valid subset $J \subseteq [n]$, the allocation $Y^i$ satisfies the requirements for an $\alpha$-EPMMS (if $|J| = 2$) and an $\alpha$-EGMMS (if $|J|$ is arbitrary) certificate.
\end{proof}

\section{Non-Degenerate Instances}

\subsection{The Envy Cycle Elimination Does Not Work in Degenerate Instances}\label{sec:envy_cycles_degenerate}

We show that the non-degeneracy assumption on the input valuations to Algorithm~\ref{alg:envy_cycles} is necessary; without it, the $4/5$-PMMS guarantee fails. We first show that without additional assumptions on valuations (beyond identical rankings), Envy Cycle Elimination can guarantee at most $2/3$-PMMS; this example includes zero-valued items. We then show that ruling out zero-valued items is still insufficient to recover the $4/5$-PMMS guarantee: we provide an example without zero-valued items in which Envy Cycle Elimination can guarantee at most $3/4$-PMMS. 

We begin with the first example, which permits zero-valued items.

\begin{observation}
Without the assumption that the valuations are non-degenerate, the Envy Cycle Elimination algorithm (Algorithm~\ref{alg:envy_cycles}) does not guarantee better than $2/3$-PMMS on additive valuations with identical rankings.
\end{observation}
\begin{proof}
Fix $\epsilon > 0$, and consider an instance with $3$ agents and $7$ items, with the following additive valuations:
\begin{center}
\begin{tabular}{c c c c c c c c}
  & $g_1$ & $g_2$ & $g_3$ & $g_4$ & $g_5$ & $g_6$ & $g_7$ \\
 \text{Agent 1} & 2 & 2 & 2 & 1 & 1 & $\epsilon$ & 0 \\
 \text{Agent 2} & 2 & 2 & 2 & 1 & 1 & $\epsilon$ & 0 \\
 \text{Agent 3} & 1 & 0 & 0 & 0 & 0 & 0 & 0
\end{tabular}
\end{center}
Note that all three agents share the same ranking $g_1, g_2,  \ldots,  g_7$, so this is an identical-rankings instance. Consider the following execution of Algorithm~\ref{alg:envy_cycles}:
\begin{enumerate}
    \item Agent $1$ receives item $g_1$. Agent $1$ is unenvied before this assignment since her bundle is empty.
    \item Agent $2$ receives items $g_2$ and $g_3$. Agent $2$ is unenvied before each of these assignments: initially her bundle is empty, and after receiving $g_2$, agent $1$ values $\{g_2\}$ at $2$ (equal to her own bundle) and agent $3$ values $\{g_2\}$ at $0$, so neither envies agent $2$.
    \item Agent $3$ receives items $g_4$, $g_5$, and $g_6$. Before each assignment, agent $3$ is unenvied since agents $1$ and $2$ value any subset of $\{g_4, g_5, g_6\}$ at no more than their own bundles ($2$ and $4$, respectively).
\end{enumerate}
At this point, the envy graph, whose nodes are the agents and which contains an edge $i \to j$ whenever $v_i(X_j) > v_i(X_i)$, is as follows:
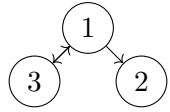
\begin{center}
\begin{tikzpicture}[main/.style = {draw, circle}]
\node[main] (1) {$1$};
\node[main] (2) [below right of=1] {$2$};
\node[main] (3) [below left of=1] {$3$};
\draw[->] (1) -- (2);
\draw[->] (1) -- (3);
\draw[->] (3) -- (1);
\end{tikzpicture}
\end{center}
No agent has in-degree $0$, so the algorithm cannot proceed without resolving the cycle $1 \to 3 \to 1$, which it does by swapping the bundles of agents $1$ and $3$. After the swap, agent $1$ holds $\{g_4, g_5, g_6\}$, valued at $2 + \epsilon$, while agent $2$ holds $\{g_2, g_3\}$. We calculate:
\begin{align*}
    \mu_1(2, \{g_2, g_3, g_4, g_5, g_6\}) \geq \min\{ 2 + 1 + \epsilon , 2 + 1 \} = 3.
\end{align*}
As $\epsilon \to 0$, the ratio $\frac{2 + \epsilon}{3} \to \frac{2}{3}$, completing the proof.
\end{proof}

We now turn to the second example, which shows that excluding zero-valued items still does not suffice to recover the $4/5$-PMMS guarantee.

\begin{observation}
    Even when all item values are strictly positive, the Envy Cycle Elimination algorithm (Algorithm~\ref{alg:envy_cycles}) does not guarantee better than $3/4$-PMMS on additive valuations with identical rankings.
\end{observation}

\begin{proof}
Assume there are three agents and seven items. Let $\epsilon > 0$ be a sufficiently small constant. The valuation functions are given by:

\begin{center}
\begin{tabular}{cccccccc}
 & $g_1$ & $g_2$ & $g_3$ & $g_4$ & $g_5$ & $g_6$ & $g_7$ \\
Agent $1$ & $3$ & $3$ & $2$ & $2$ & $1$ & $\epsilon$ & $\epsilon$ \\
Agent $2$ & $3$ & $3$ & $2$ & $2$ & $1$ & $\epsilon$ & $\epsilon$ \\
Agent $3$ & $10$ & $1$ & $1$ & $1$ & $1$ & $1$ & $\epsilon$ \\
\end{tabular}
\end{center}

Note that all three agents share the same ranking $g_1, g_2,  \ldots,  g_7$, so this is an identical-rankings instance. Consider the following execution of Algorithm~\ref{alg:envy_cycles}:

After the first three iterations, every agent holds an item:
\[
\begin{aligned}
X_1 &= \{g_1\}, & v_1(X_1) &= 3, \ \ \ \ \ \ \ \ \ \ \ \ \ \ \\
X_2 &= \{g_2\}, & v_2(X_2) &= 3, \\
X_3 &= \{g_3\}, & v_3(X_3) &= 1
\end{aligned}
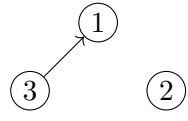
\quad \quad
\begin{tikzpicture}[
    baseline={(current bounding box.center)},
    scale=1.2,
    main/.style={draw,circle,inner sep=2pt}
]
\node[main] (1) {$1$};
\node[main] (2) [below right=0.75cm of 1] {$2$};
\node[main] (3) [below left=0.75cm of 1] {$3$};
\draw[->] (3)--(1);
\end{tikzpicture}
\]

Next, agent 2 receives an item:
\[
\begin{aligned}
X_1 &= \{g_1\}, & v_1(X_1) &= 3, \\
X_2 &= \{g_2,g_4\}, & v_2(X_2) &= 5, \ \ \ \ \ \ \ \ \ \ \  \\
X_3 &= \{g_3\}, & v_3(X_3) &= 1
\end{aligned}
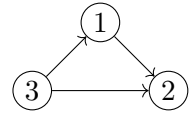
\quad \quad
\begin{tikzpicture}[
    baseline={(current bounding box.center)},
    scale=1.2,
    main/.style={draw,circle,inner sep=2pt}
]
\node[main] (1) {$1$};
\node[main] (2) [below right=0.75cm of 1] {$2$};
\node[main] (3) [below left=0.75cm of 1] {$3$};
\draw[->] (3)--(1);
\draw[->] (3)--(2);
\draw[->] (1)--(2);
\end{tikzpicture}
\]

Next, agent 3 receives an item:
\[
\begin{aligned}
X_1 &= \{g_1\}, & v_1(X_1) &= 3, \\
X_2 &= \{g_2,g_4\}, & v_2(X_2) &= 5, \\
X_3 &= \{g_3,g_5\}, & v_3(X_3) &= 2 \ \ \ \ \ \ \ \ \ \ \ \
\end{aligned}
\quad \quad
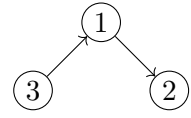
\begin{tikzpicture}[
    baseline={(current bounding box.center)},
    scale=1.2,
    main/.style={draw,circle,inner sep=2pt}
]
\node[main] (1) {$1$};
\node[main] (2) [below right=0.75cm of 1] {$2$};
\node[main] (3) [below left=0.75cm of 1] {$3$};
\draw[->] (3)--(1);
\draw[->] (1)--(2);
\end{tikzpicture}
\]

Next, agent 3 receives another item:
\[
\begin{aligned}
X_1 &= \{g_1\}, & v_1(X_1) &= 3, \ \ \ \ \ \ \ \\
X_2 &= \{g_2,g_4\}, & v_2(X_2) &= 5, \\
X_3 &= \{g_3,g_5,g_6\}, & v_3(X_3) &= 3
\end{aligned}
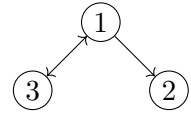
\quad \quad
\begin{tikzpicture}[
    baseline={(current bounding box.center)},
    scale=1.2,
    main/.style={draw,circle,inner sep=2pt}
]
\node[main] (1) {$1$};
\node[main] (2) [below right=0.75cm of 1] {$2$};
\node[main] (3) [below left=0.75cm of 1] {$3$};
\draw[<->] (3)--(1);
\draw[->] (1)--(2);
\end{tikzpicture}
\]

At this point, a cycle is created and we resolve it:
\[
\begin{aligned}
X_1 &= \{g_3,g_5,g_6\}, & v_1(X_1) &= 3 + \epsilon, \\
X_2 &= \{g_2,g_4\}, & v_2(X_2) &= 5, \\
X_3 &= \{g_1\}, & v_3(X_3) &= 10
\end{aligned}
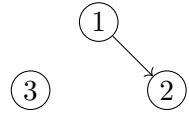
\quad \quad
\begin{tikzpicture}[
    baseline={(current bounding box.center)},
    scale=1.2,
    main/.style={draw,circle,inner sep=2pt}
]
\node[main] (1) {$1$};
\node[main] (2) [below right=0.75cm of 1] {$2$};
\node[main] (3) [below left=0.75cm of 1] {$3$};
\draw[->] (1)--(2);
\end{tikzpicture}
\]

Finally, suppose agent 3 receives the last item. The algorithm terminates with this allocation.

Observe that:
\[
v_1(X_1) = 3 + \epsilon,
\]
while
\[
\mu_1(2, X_1 \cup X_2)
= \min \big( v_1(\{g_2,g_5\}),\; v_1(\{g_3,g_4,g_6\}) \big)
= \min\{4,4+\epsilon\} = 4.
\]

Thus, agent 1 cannot obtain more than $\tfrac{3+\epsilon}{4} \to \tfrac{3}{4}$ of her pairwise maximin share as $\epsilon \to 0$, completing the proof.
\end{proof}

\subsection{Reduction to Non-Degenerate Instances}

We adopt the following definition of non-degenerate instances from \citet*{ChaudhuryGM24}.
 
\begin{definition}[Non-Degenerate Instances, \cite{ChaudhuryGM24}]
An instance is called \emph{non-degenerate} if no agent assigns equal value to two distinct sets; that is, for every $i \in N$, $v_i(S) \neq v_i(T)$ whenever $S \neq T$.
\end{definition}
 
We now show that it suffices to restrict attention to non-degenerate instances.
 
In the following proposition, we show that to establish Proposition~\ref{prop:4_5_PMMS_identical_rankings_1}, it is enough to consider non-degenerate instances. \citet*{ChaudhuryGM24} proved an analogous reduction for EFX, and we adapt their argument to our setting.
\begin{proposition}[Reduction to Non-Degenerate Instances]\label{prop:nd_wlog}
If a polynomial-time algorithm exists for computing a 4/5-PMMS and EFX allocation on every instance with \emph{non-degenerate} additive valuations and identical rankings, then it can be extended to handle every instance with (possibly degenerate) additive valuations and identical rankings.
\end{proposition}

\begin{proof}
Fix $\alpha \le 1$, and let $M = \{g_1, g_2, \dots, g_m\}$ be ordered from least to most valuable. Given any instance $\mathcal{I}$ with valuation profile $(v_1, \ldots, v_n)$, we construct a perturbed instance $\mathcal{I}(\varepsilon)$ with valuation profile $(v'_1, \ldots, v'_n)$, where each $v'_i$ is defined by
\[
v_i'(g_j) = v_i(g_j) + \varepsilon \cdot 2^j.
\]
 
For any $k < \ell$ and any $i \in N$, we have
\[
v'_i(g_k) = v_i(g_k) + \varepsilon 2^k < v_i(g_\ell) + \varepsilon 2^\ell = v'_i(g_\ell).
\]
Hence $\mathcal{I}(\varepsilon)$ preserves the original ordering of items for every agent, so the perturbed instance also has identical rankings.

For every $i \in N$, define 
\[
\delta_1(i) = \min_{S,T : \alpha v_i(S) \neq v_i(T)} |\alpha v_i(S) - v_i(T)|, \quad
\delta_2(i) = \min_{S,T : v_i(S) \neq v_i(T)} |v_i(S) - v_i(T)|.
\]
Set $\delta_1 = \min_i \delta_1(i)$, $\delta_2 = \min_i \delta_2(i)$, and $\delta = \min(\delta_1, \delta_2)$. Choose $\varepsilon > 0$ small enough that 
\begin{align*}
    \varepsilon \cdot 2^{m+1} < \delta.
\end{align*}

We first show that for any agent $i$ and any $S, T \subseteq M$,
\begin{equation} \label{order_preserving}
     v_i(S) > v_i(T) \Longrightarrow  v'_i(S) > v'_i(T).
\end{equation}
Suppose $v_i(S) > v_i(T)$. Then
\begin{align*}
v'_i(S) - v'_i(T)
&= v_i(S) - v_i(T) + \varepsilon \!\left(\sum_{g_j \in S} 2^j - \sum_{g_j \in T} 2^j\right) \\
&\ge \delta_2 - \varepsilon \sum_{g_j \in T} 2^j \\
&\ge \delta - \varepsilon (2^{m+1} - 1) \\
&> 0.
\end{align*}
\newline
By the same argument, for any agent $i$ and any $S, T \subseteq M$,
\begin{equation} \label{order_preserving_alpha}
    \alpha v_i(S) > v_i(T) \Longrightarrow \alpha v'_i(S) > v'_i(T).
\end{equation}
Indeed, if $\alpha v_i(S) > v_i(T)$, then
\begin{align*}
\alpha v'_i(S) - v'_i(T)
&= \alpha v_i(S) - v_i(T) + \varepsilon \!\left(\sum_{g_j \in S} \alpha 2^j - \sum_{g_j \in T} 2^j\right) \\
&\ge \delta_1 - \varepsilon \sum_{g_j \in T} 2^j \\
&\ge \delta - \varepsilon (2^{m+1} - 1) \\
&> 0.
\end{align*}

We next verify that $\mathcal{I}(\varepsilon)$ is non-degenerate. Take any two distinct sets $S, T \subseteq M$ and any agent $i \in N$. If $v_i(S) \neq v_i(T)$, then $v'_i(S) \neq v'_i(T)$ by Equation~\eqref{order_preserving}. If instead $v_i(S) = v_i(T)$, then
\[
v'_i(S) - v'_i(T) = \varepsilon \!\left(\sum_{g_j \in S \setminus T} 2^j - \sum_{g_j \in T \setminus S} 2^j\right) \neq 0,
\]
since $S \neq T$ and the sums of distinct subsets of powers of two differ. Hence $\mathcal{I}(\varepsilon)$ is non-degenerate.

It remains to show that any $\alpha$-PMMS allocation in $\mathcal{I}(\varepsilon)$ is also an $\alpha$-PMMS allocation in $\mathcal{I}$. Suppose, for contradiction, that $X$ is an $\alpha$-PMMS allocation in $\mathcal{I}(\varepsilon)$ but not in $\mathcal{I}$. Then there exist agents $i, j$ such that $\alpha \mu_i(2, X_i \cup X_j) > v_i(X_i)$, so there is a partition $(A_1, A_2)$ of $X_i \cup X_j$ satisfying $\alpha v_i(A_1) > v_i(X_i)$ and $\alpha v_i(A_2) > v_i(X_i)$. Applying Equation~\eqref{order_preserving_alpha}, we obtain $\alpha v'_i(A_1) > v'_i(X_i)$ and $\alpha v'_i(A_2) > v'_i(X_i)$, which gives
\[
\alpha \mu_{v'_i}(2, X_i \cup X_j)
  \ge \alpha \min(v'_i(A_1), v'_i(A_2))
  > v'_i(X_i),
\]
contradicting the assumption that $X$ is an $\alpha$-PMMS allocation in $\mathcal{I}(\varepsilon)$.

Therefore, whenever every non-degenerate instance admits an $\alpha$-PMMS allocation, any instance $\mathcal{I}$ can be handled by passing to its non-degenerate perturbation $\mathcal{I}(\varepsilon)$, and any $\alpha$-PMMS allocation in $\mathcal{I}(\varepsilon)$ is also an $\alpha$-PMMS allocation in $\mathcal{I}$.
\end{proof}

\end{document}